\documentclass[pra,twocolumn, showpacs, preprintnumbers,amsmath,amssymb]{revtex4}

\usepackage{graphicx, color, graphpap}
\usepackage{amsmath}
\usepackage{enumitem}
\usepackage{amssymb}
\usepackage{amsthm}
\usepackage{pstricks}
\usepackage{float}
\usepackage{multirow}
\usepackage{hyperref}

\long\def\ca#1\cb{} 

\newcommand{\ket}[1]{|#1\rangle}               
\newcommand{\colo}{\,\hbox{:}\,}              
\newcommand{\bra}[1]{\langle #1|}              
\newcommand{\dya}[1]{\ket{#1}\!\bra{#1}}
\newcommand{\dyad}[2]{\ket{#1}\!\bra{#2}}        
\newcommand{\fid}{\text{fid}}
\newcommand{\rhot}{\tilde{\rho}}
\newcommand{\Sbf}{\mathbf{S}}

\newcommand{\Ebb}{\mathbb{E}}

\newcommand{\Fbar}{\overline{F}}

\newcommand{\sgbar}{\overline{\sigma}}
\newcommand{\taubar}{\overline{\tau}}
\newcommand{\Khat}{\widehat{K}}

\newcommand{\BC}{\mathcal{B}}

\newcommand{\DC}{\mathcal{D}}
\newcommand{\EC}{\mathcal{E}}
\newcommand{\FC}{\mathcal{F}}

\newcommand{\HC}{\mathcal{H}}

\newcommand{\VC}{\mathcal{V}}
\newcommand{\WC}{\mathcal{W}}
\newcommand{\XC}{\mathcal{X}}
\newcommand{\YC}{\mathcal{Y}}

\newcommand{\Tr}{{\rm Tr}}

\renewcommand{\geq}{\geqslant}
\renewcommand{\leq}{\leqslant}

\newcommand{\MQ}{\textsf{MQ}}
\newcommand{\mQ}{\textsf{mQ}}

\newcommand{\ot}{\otimes}
\newcommand{\ad}{^\dagger}

\newcommand*{\id}{\openone}

\newcommand{\al}{\alpha }

\newcommand{\dl}{\delta }
\newcommand{\Dl}{\Delta}
\newcommand{\ep}{\epsilon}

\newcommand{\lm}{\lambda }
\newcommand{\Lm}{\Lambda }

\newcommand{\sg}{\sigma }

\newcommand{\om}{\omega }

\newtheoremstyle{example}{\topsep}{\topsep}%
{}
{}
{\bfseries}
{.}
{   }
{\thmname{#1}\thmnumber{ #2}}
\theoremstyle{example}

\newtheorem{theorem}{Theorem}
\newtheorem{lemma}[theorem]{Lemma}

\theoremstyle{definition}

\begin{document}

\title{Collapse of the quantum correlation hierarchy links entropic uncertainty to entanglement creation}

\author{Patrick J. Coles}
\affiliation{Department of Physics, Carnegie Mellon University, Pittsburgh, Pennsylvania 15213, USA}
\affiliation{Centre for Quantum Technologies, National University of Singapore, Singapore}

\begin{abstract}
Quantum correlations have fundamental and technological interest, and hence many measures have been introduced to quantify them. Some hierarchical orderings of these measures have been established, e.g.\ discord is bigger than entanglement, and we present a class of bipartite states, called premeasurement states, for which several of these hierarchies collapse to a single value. Because premeasurement states are the kind of states produced when a system interacts with a measurement device, the hierarchy collapse implies that the uncertainty of an observable is quantitatively connected to the quantum correlations (entanglement, discord, etc.)\ produced when that observable is measured. This fascinating connection between uncertainty and quantum correlations leads to a reinterpretation of entropic formulations of the uncertainty principle, so-called entropic uncertainty relations, including ones that allow for quantum memory. These relations can be thought of as lower-bounds on the entanglement created when incompatible observables are measured. Hence, we find that entanglement creation exhibits \emph{complementarity}, a concept that should encourage exploration into ``entanglement complementarity relations". 
\end{abstract}


\maketitle

\section{Introduction}

As researchers attempt to develop the ultimate theory of information, encompassing both classical and quantum information, it is becoming increasingly apparent that quantum correlations - correlations that go beyond classical correlations - are of great fundamental and technological interest. Questions like, what gives the quantum advantage in computing tasks \cite{DatShaCav08}, have motivated the definition and study of many quantitative measures of quantum correlations, ranging from entanglement \cite{HHHH09} to discord \cite{OllZur01} and other related measures \cite{ModiEtAl2011review}. Some of these measures are operationally motivated, e.g.\ the number of Einstein-Podolsky-Rosen (EPR) pairs that can be distilled from the state, others are geometrically motivated like the distance to the nearest separable state or the nearest classical state, while others are motivated due to their ease of calculation. The zoo of quantum correlation measures is vast, and yet the story is simple for bipartite pure states, where the entropy of the reduced state pretty much captures it all. While it would be nice if the correlations of mixed states shared the simplicity of those of pure states, in general, we must settle for a hierarchical ordering of the various correlation measures, e.g., discord is bigger than entanglement \cite{PianiAdessoPRA.85.040301, HorEtAl05}, which in turn is bigger than coherent information \cite{DevWin05}.

In the present article, we consider a class of bipartite states for which this zoo dramatically simplifies to a single number; various quantum correlation measures which are in general related by a hierarchy of \textit{inequalities} become equal for these states, so we say that these states ``collapse the quantum correlation hierarchy". Hence these states are like pure states in that their correlations are ``simple", even though the set includes not only pure states but also some mixed states. Interestingly, the set of states that collapse the quantum correlation hierarchy corresponds precisely to the set of states that can be produced when a system interacts with a measurement device. These states have been called premeasurement states, since the unitary interaction (called premeasurement) that potentially correlates the system to the measurement device is the first step in the measurement process \cite{ZurekReview}. The fact that premeasurement states collapse the quantum correlation hierarchy has significant consequences, and much of this article is devoted to exploring these consequences. 

The most interesting consequence is a connection to uncertainty and the uncertainty principle. While the study of quantum correlations has seen a revolution of sorts recently, so has the study of the uncertainty principle. In quantitative expressions of the uncertainty principle, so-called uncertainty relations, researchers have replaced the standard deviation, the uncertainty measure employed in the original formulations \cite{Heisenberg, Robertson}, with \emph{entropy} measures, leading to a variety of different entropic uncertainty relations (EURs) \cite{EURreview1}, which are more readily applied to information-processing tasks. Allowing the observer to possess ``quantum memory" (a quantum system that may be entangled to the system of interest) has led to EURs \cite{RenesBoileau, BertaEtAl, TomRen2010, ColesEtAl, ColesColbeckYuZwolak2012PRL} with direct application in entanglement witnessing \cite{LXXLG,PHCFR} and cryptography \cite{TLGR}.

Our results allow us to establish a precise and general connection between the uncertainty of an observable and the quantum correlations, such as entanglement, created when that observable is measured (more precisely, premeasured). As a consequence, a wide variety of EURs, including those allowing for quantum memory, are subject to reinterpretation. The conventional interpretation is that EURs are lower bound on our inability to predict the outcomes of incompatible measurements, but our results imply that EURs can be thought of as lower bounds on the \textit{entanglement created} in incompatible measurements.

It is helpful to illustrate this connection with a simple example. Consider a qubit in state $\ket{0}$, then the unitary associated with a $Z$-measurement is a controlled-not (CNOT) acting on a register qubit that is initially in state $\ket{0}$. In this case, the overall state evolves trivially: $\ket{0} \ket{0}\to \ket{0} \ket{0}$, producing no entanglement. But if instead we did an $X$-measurement, with a CNOT controlled by the $\{\ket{+}, \ket{-}\}$ basis, then the state evolves as $\ket{0} \ket{0}= (\ket{+}+\ket{-}) \ket{0}/\sqrt{2} \to (\ket{+} \ket{0} +\ket{-} \ket{1})/\sqrt{2} $, which is maximally entangled. Note that the uncertainty of the $Z$ ($X$) observable was zero (maximal), which is connected to the final entanglement being zero (maximal). This example shows the connection of uncertainty to entanglement creation, and it also shows the \emph{complementarity} of entanglement creation: the $X$ measurement must create entanglement because the $Z$ measurement does not.

We remark that the entanglement created in measurements has been an area of interest previously \cite{ZurekReview, VedralPRL2003}, and there is renewed interest in this as it provides a general framework for quantifying discord \cite{PianiEtAl11, StrKamBru11, PianiAdessoPRA.85.040301}. It should, therefore, be of interest that our reinterpretation of EURs implies that the entanglement (and discord) created in measurements exhibits complementarity. This idea, which seems to be a general principle, suggests that there are classes of inequalities that capture the complementarity of quantum mechanics, which have yet to be explored and involve entanglement (or discord) creation. There is generally a trade-off; for a given quantum state, if one avoids creating quantum correlations in one measurement, then a complementary measurement will necessarily create such correlations.

In summary, we emphasize three main concepts in this article: (1) the quantum correlation hierarchy dramatically simplifies for premeasurement states, (2) an observable's uncertainty quantifies the entanglement created upon measuring that observable, and (3) entanglement creation exhibits complementarity. Mathematically speaking, concept (1) implies concept (2) which in turn implies concept (3), as we will discuss.

The rest of the manuscript is organized as follows. In Section~\ref{sct22} we define various classes of bipartite quantum states, including premeasurement states. In Section~\ref{sct5} we consider several different quantum correlation hierarchies, and we show that premeasurement states collapse these hierarchies. In particular, we consider hierarchies of measures based on a generic relative entropy, measures related to the von Neumann entropy, and measures related to smooth entropies. In Section~\ref{sct6}, we use these results to connect an observable's uncertainty to the quantum correlations created when that observable is measured. Then we argue that this gives a reinterpretation for EURs in Section~\ref{sct7}, focusing particularly on the complementarity of entanglement creation. Section~\ref{sct8} gives a few more implications of our results and discusses the future outlook for ``entanglement complementarity relations". Section~\ref{sct9} gives some concluding remarks.

\section{Classes of bipartite states}\label{sct22}

\subsection{Classical, separable, and entangled states}\label{sct22a}

Since we will be considering various correlation measures, it is helpful to define particular classes of bipartite quantum states. First, consider the set of all separable states, hereafter denoted $\textsf{Sep}$, which have the general form of a convex combination of tensor products:
\begin{equation}
\label{eqn1}
\rho_{AB}= \sum_j p_j \rho_{A,j} \ot \rho_{B,j},
\end{equation}
where $\{p_j\}$ is some probability distribution and $\rho_{A,j}$ and $\rho_{B,j}$ are density operators on systems $A$ and $B$. Entangled states are defined as those states that are not separable; we denote this set as $\textsf{Ent}$, the complement of $\textsf{Sep}$.

A special kind of separable state is a classical state, often called a classical-classical or $\textsf{CC}$ state, with the general form:
\begin{equation}
\label{eqn2}
\rho_{AB}= \sum_{j,k} p_{j,k} \dya{j} \ot \dya{k},
\end{equation}
which is like the embedding of a classical joint probability distribution $\{p_{j,k}\}$ in a Hilbert space, where $\{\ket{j}\}$ and $\{\ket{k}\}$ are orthonormal bases on $\HC_A$ and $\HC_B$, respectively. More generally, a state can be classical with respect to one of the subsystems, e.g., of the form:
\begin{equation}
\label{eqn3}
\rho_{AB}= \sum_{j} p_{j} \dya{j} \ot \rho_{B,j},
\end{equation}
in which case it is called classical-quantum or $\textsf{CQ}$, and naturally is called quantum-classical or $\textsf{QC}$ if it is classical with respect to system $B$. The following relations between these sets should be clear from the above definitions:
\begin{equation}
\label{eqn4}
\textsf{CQ}\subset \textsf{Sep}, \quad \textsf{QC}\subset \textsf{Sep}, \quad \textsf{CQ} \cap \textsf{QC} = \textsf{CC},
\end{equation}
and a Venn diagram in Fig.~\ref{fgr2} depicts these relations.

\subsection{Pure states}\label{sct22aa}

Pure states can be either separable or entangled, though if a pure state is separable, it is necessarily a classical state (more specifically, a product state), in other words,
\begin{equation}
\label{eqn5}
(\textsf{Pure} \cap \textsf{Sep}) \subset \textsf{CC},
\end{equation}
as depicted in Fig.~\ref{fgr2}. The correlations of pure states are very well understood, e.g., see \cite{HHHH09,NieChu00}, and one of our contributions is to characterize a set of states whose correlations are somewhat analogous to those of pure states, a set that encompasses, but goes beyond, pure states. We discuss this set below.

\begin{figure}
\begin{center}
\includegraphics[width=2in]{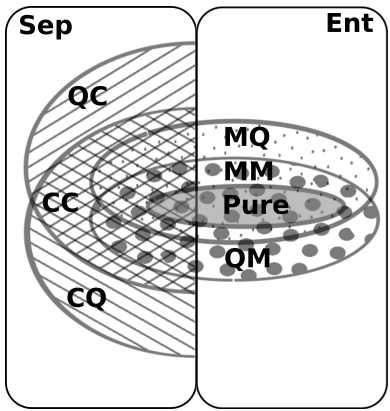}
\caption{Venn diagram for several classes of bipartite states. Bipartite states are either separable ($\textsf{Sep}$) or non-separable ($\textsf{Ent}$). Subsets of $\textsf{Sep}$ include $\textsf{QC}$ and $\textsf{CQ}$, which are shaded with lines slanted up-to-the-right and up-to-the-left, respectively, and $\textsf{CC}$ is the intersection of these two sets. The set of pure states is shaded solid gray and is contained inside $\textsf{MM}$, the intersection of $\textsf{MQ}$ and $\textsf{QM}$, which are respectively shaded with small dots and large dots. Note: the figure is not to scale, and is only meant to convey the set relationships given in Eqs.~\eqref{eqn4}, \eqref{eqn5}, and \eqref{eqn7}--\eqref{eqn9}.}
\label{fgr2}
\end{center}
\end{figure}

\subsection{Premeasurement states}\label{sct22b}

Consider the following set of bipartite states:
\begin{equation}
\label{eqn6}
\textsf{MQ}:=\{\rho_{AB} : \rho_{AC}\in \textsf{CQ}\text{ for pure }\rho_{ABC} \}.
\end{equation}
Here, $\rho_{ABC}$ is any purification of $\rho_{AB}$, so we are considering the set of states $\rho_{AB}$ such that there exists a purification $\rho_{ABC}$ whose marginal $\rho_{AC}$ is of the general form of \eqref{eqn3}, i.e., classical with respect to system $A$. (If $\rho_{AC}\in \textsf{CQ}$ for some purification of $\rho_{AB}$, then the same will be true for all other purifications.) If $A$ and $B$ change roles in \eqref{eqn6}, i.e., if $\rho_{BC} \in \textsf{CQ}$, then we denote this set as $\textsf{QM}$, and if both $\rho_{AC}$ and $\rho_{BC}$ are $\textsf{CQ}$, then we say that $\rho_{AB}\in \textsf{MM}$. In other words, 
\begin{equation}
\label{eqn7}
\textsf{MQ} \cap \textsf{QM} = \textsf{MM}.
\end{equation}
It turns out, as we will see below, that $\textsf{MM}$ corresponds precisely to the ``maximally correlated states", introduced by Rains \cite{RainsPhysRevA.60.179}.

It is clear that if $\rho_{AB}$ is pure, then any purifying system $C$ will be in a tensor product with (uncorrelated with) $AB$, hence both marginals $\rho_{AC}$ and $\rho_{BC}$ will be classical. So all pure states are in $\textsf{MM}$,
\begin{equation}
\label{eqn8}
\textsf{Pure}\subset \textsf{MM}.
\end{equation}

Figure~\ref{fgr2} schematically depicts Eqs.~\eqref{eqn7} and \eqref{eqn8}. Also captured by this figure is an extension of \eqref{eqn5} to $\textsf{MQ}$ and $\textsf{QM}$ states: 
\begin{equation}
\label{eqn9}
(\textsf{MQ} \cap \textsf{Sep}) \subset \textsf{CC},\quad (\textsf{QM} \cap \textsf{Sep}) \subset \textsf{CC}.
\end{equation}
While \eqref{eqn9} is not at all obvious, it is a consequence of our results proven in Sec.~\ref{sct5}.

Our curious notation $\textsf{MQ}$ is motivated by the fact that one of the subsystems, namely system $A$ in \eqref{eqn6}, is behaving like a measurement device in a way that we elaborate on below. Thus, one can read $\textsf{MQ}$ as ``measurement device - quantum", analogous to how one reads $\textsf{CQ}$ as ``classical - quantum".

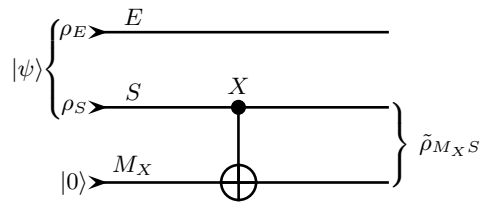
\begin{figure}[t]
\begin{center}

\begin{pspicture}(-1.0,-0.3)(5.5,2.4) 
\newpsobject{showgrid}{psgrid}{subgriddiv=1,griddots=10,gridlabels=6pt}
\def\lwd{0.035} 
\def\lwb{0.10}  
\psset{
labelsep=2.0,
arrowsize=0.150 1,linewidth=\lwd}
\def\circ{%
\pscircle[fillcolor=white,fillstyle=solid]{0.4}}
\def\rectg(#1,#2,#3,#4){
\psframe[fillcolor=white,fillstyle=solid](#1,#2)(#3,#4)}
\def\rectc(#1,#2){%
\psframe[fillcolor=white,fillstyle=solid](-#1,-#2)(#1,#2)}
\def\brecs(#1,#2){%
\psframe[fillcolor=white,fillstyle=solid,linestyle=none](-#1,-#2)(#1,#2)}
\def\squ{\rectc(0.35,0.35)}
\def\lbrk#1{\left\{\vrule height #1cm depth #1cm width 0pt\right.}
\def\rbrk#1{\left.\vrule height #1cm depth #1cm width 0pt\right\}}
\def\vvv#1{\vrule height #1 cm depth #1 cm width 0pt}
\def\rbrac{$\left.\vvv{0.5}\right\}$}
\def\lbrac{$\left\{\vvv{0.6}\right.$}

\def\rdet{0.35}  
\def\detect{
\psarc[fillcolor=white,fillstyle=solid](0,0){\rdet}{-90}{90}
\psline(0,-\rdet)(0,\rdet)}
\def\ddetect{
\psarc[fillcolor=white,fillstyle=solid](0.2,0){\rdet}{-90}{90}
\rectg(0,-\rdet,0.2,\rdet)}
\def\delx{0.8} 
\def\edetect{
\psarc[fillcolor=white,fillstyle=solid](\delx,0){\rdet}{-90}{90}
\psframe[fillcolor=white,fillstyle=solid,linestyle=none](0,-\rdet)(\delx,\rdet)
\psline(\delx,\rdet)(0,\rdet)(0,-\rdet)(\delx,-\rdet)}
\def\rdot{0.10} \def\rodot{0.12} 
\def\dot{\pscircle*(0,0){\rdot}} 
\def\odot{\pscircle[fillcolor=white,fillstyle=solid](0,0){\rodot}} 
\def\dput(#1)#2#3{\rput(#1){#2}\rput(#1){#3}} 
\def\rcp{0.25}  
\def\cnt{\pscircle(0,0){\rcp}\psline(-\rcp,0)(\rcp,0)\psline(0,-\rcp)(0,\rcp)}
\def\cnotg(#1,#2,#3){%
\psline(#1,#2)(#1,#3)\rput(#1,#2){\cnt}\pscircle*(#1,#3){\rdot}}
\def\cpxu{0.15} 
\def\cpx{\psline(-\cpxu,-\cpxu)(\cpxu,\cpxu)%
\psline(-\cpxu,\cpxu)(\cpxu,-\cpxu)} 
\def\cphase(#1,#2,#3){
\psline(#1,#2)(#1,#3)\rput(#1,#2){\cpx}\rput(#1,#3){\cpx}} 
\def\csqug(#1,#2,#3){%
\psline(#1,#2)(#1,#3)\rput(#1,#2){\squ}\pscircle*(#1,#3){\rdot}}
\def\csquug(#1,#2,#3)#4{%
\psline(#1,#2)(#1,#3)\dput(#1,#2){\squ}{#4}\pscircle*(#1,#3){\rdot}}
\def\lwdsh{0.020} 
\def\tmline(#1,#2,#3)#4{
\rput(#1,#2){\psline[linestyle=dashed,linewidth=\lwdsh](0.0,-0.3)(0.0,#3)}
\rput(0,-0.4){\rput[t](#1,#2){#4}}}
\def\tnline(#1,#2,#3)#4{
\rput(#1,#2){\psline[linestyle=dashed,linewidth=\lwdsh](0.0,-0.3)(0.0,#3)}
\rput(0,-0.7){\rput[B](#1,#2){#4}}}
\def\vertdash{\psline[linestyle=dashed,linewidth=\lwdsh](0.0,-0.3)(0.0,1.5)}
\psline{>-}(0.0,2.0)(4.0,2.0)
\psline{>-}(0.0,0.0)(4.0,0.0)
\psline{>-}(0.0,1.0)(4.0,1.0)

\cnotg(2.0,0.0,1.0)

\rput(-0.45,1.5){\lbrac}
\rput(4.12,0.5){\rbrac}

\rput[r](-0.58,1.5){$\ket{\psi} $} 
\rput[r](0.0,1.0){$\rho_S $} 
\rput[r](0.0,2.0){$\rho_E $} 
\rput[r](0.0,0.0){$\ket{0} $} 
\rput[r](5.23,0.5){$\tilde{\rho}_{M_XS} $} 
\rput[b](0.6,2.1){$E$} 
\rput[b](0.6,1.1){$S$} 
\rput[b](0.6,0.1){$M_X $} 
\rput[b](2.0,1.18){$X$} 
\end{pspicture}

\caption{%
The interaction of $S$ with an $X$-measurement device is modeled as a generalized CNOT, controlled by a PVM $X$ on $S$, as given by Eq.~\eqref{eqn10}. System $E$ purifies $\rho_S$. 
\label{fgr1}}
\end{center}
\end{figure}

To make the connection to measurement, it is helpful to switch to a more intuitive notation for the various subsystems. We consider the interaction of system $S$ with a device $M_X$ that measures observable $X=\{X_j\}$ of $S$, where the $X_j$ are orthogonal projectors that sum to the identity on $\HC_S$. [We emphasize that the $X_j$ are not necessarily rank-one; $X$ is a general projection valued measure (PVM).] This can be modeled by considering a set of orthonormal states $\{\ket{j}\}$ on $ M_X$, and if $S$ is hit by projector $X_j$, then $M_X$ goes to the state $\ket{j}$, as follows:
\begin{equation}
\label{eqn10}
\ket{0}_{M_X} \ket{\psi}_S \to \sum_j \ket{j}_{M_X} (X_j \ket{\psi})_S =V_X\ket{\psi}_S,
\end{equation}
which is essentially a controlled-shift operation, and the notation is simplified by defining the isometry:
\begin{equation}
\label{eqn11}
V_X = \sum_j   \ket{j}_{M_X} \ot (X_j)_S.
\end{equation}
In \eqref{eqn10} we assumed that both $S$ and $M_X$ were initially described by pure states. More generally either state could be mixed, although we could always lump the measurement device's environment into system $M_X$ and hence purify the state of $M_X$ and call it the $\ket{0}$ state. We make this simplification throughout, although see \cite{VedralPRL2003} for a treatment allowing the measurement device to be in a mixed state. On the other hand, we find it convenient and natural to the think of the system's initial state as being a (possibly mixed) density operator $\rho_S$; then the final state after the interaction with $M_X$ is:
\begin{equation}
\label{eqn12}
\rhot_{M_X S} = V_X \rho_S V_X\ad .
\end{equation}
The circuit diagram for this process is depicted in Fig.~\ref{fgr1}, using the controlled-not (CNOT) symbol even though the process is slightly more general. Also shown is the quantum system that purifies $\rho_S$, called $E$. Because it is the first step in performing a measurement, this process has been called ``premeasurement", and the resulting states $\rhot_{M_X S}$ that are produced have been called ``premeasurement states" \cite{ZurekReview}. 

Now suppose we consider the set of \textit{all} premeasurement states, i.e., the set of all bipartite states that can be thought of as resulting from a process like that depicted in Fig.~\ref{fgr1}. It turns out that this set is precisely equivalent to $\textsf{MQ}$, as shown in Appendix~\ref{app2}. To write $\textsf{MQ}$ as the set of all premeasurement states, we revert to the notation $A$ and $B$ for the two subsystems, since we are being general and abstract again. Denote the set of all orthonormal bases $W=\{\ket{W_j}\}$ on $\HC_A$ as $\WC_A$, denote the set of all PVMs $X=\{X_j\}$ on system $B$ as $\XC_B$, and denote the set of all premeasurement isometries $V_X:\HC_B\to\HC_{AB}$ as 
\begin{equation}
\label{eqn13}
\VC=\{V_X: (\exists W \in \WC_A)( \exists X\in \XC_B)(V_X=\sum_j \ket{W_j}\ot X_j ) \}.
\end{equation}
Denoting the set of (normalized) density operators on $B$ as $\DC_B$, then we have (see Appendix~\ref{app2} for the proof)
\begin{equation}
\label{eqn14}
\textsf{MQ}=\{\rho_{AB} \colo (\exists \sg_B \in \DC_B)(\exists V_X\in \VC)(\rho_{AB}=V_X \sg_B V_X\ad) \}.
\end{equation}
In other words, the general form for states in $\textsf{MQ}$ is:
\begin{equation}
\label{eqn15}
\rho_{AB}=\sum_{j,k} \dyad{W_j}{W_k}\ot X_j\sg_BX_k
\end{equation}
for some $W\in \WC_A$, $X\in \XC_B$, and $\sg_B\in\DC_B$. It is clear from \eqref{eqn15} that if all the $X_j$ projectors are rank-one and hence $X$ can be thought of as an orthonormal basis, then the state is a ``maximally correlated state", in the sense that the $W$ basis on $A$ is perfectly correlated with the $X$ basis on $B$. So maximally correlated states are a special kind of $\textsf{MQ}$ state, corresponding to $\textsf{MM}$. But more generally, we can think of $\textsf{MQ}$ states as being one-way maximally correlated in the sense that an orthonormal basis on $A$ is perfectly correlated to some projective observable (not necessarily a basis) on $B$; again $W$ and $X$ are the two observables playing this role in \eqref{eqn15}.

We remark that $\textsf{MQ}$ is a strict subset of a set of states considered in Ref.~\cite{CorDeOFan11}, defined as follows
\begin{equation}
\label{eqn16}
\mQ :=\{\rho_{AB} : \rho_{AC}\in \textsf{Sep} \text{ for pure }\rho_{ABC} \},
\end{equation}
i.e., the set of states $\rho_{AB}$ where $\rho_{AC}$ is separable for any purification $\rho_{ABC}$. Since $\textsf{CQ}\subset \textsf{Sep}$, it is clear that $\textsf{MQ} \subset \mQ $. We believe it is important to make this connection with Ref.~\cite{CorDeOFan11}, because they showed an analog of \eqref{eqn9}, namely that $(\mQ \cap \textsf{Sep}) \subset \textsf{QC}$, a consequence of the fact that states in $\mQ $ partially collapse the quantum correlation hierarchy. However, we note in Sec.~\ref{sct5.3} that $\mQ$ states do not necessarily collapse the ``full" quantum correlation hierarchy, and that the restriction of $\mQ$ to $\textsf{MQ}$ is precisely what is needed in order to obtain the ``full" collapse.

Our notation $\textsf{mQ}$ is motivated by the following observation. Unlike $\textsf{MQ}$ there are states in $\textsf{mQ}$ of the form:
\begin{equation}
\label{eqn17}
\rho_{AB}=\sum_{j,k} \dyad{\phi_j}{\phi_k}\ot X_j\sg_BX_k = \tilde{V}_X \sg_B \tilde{V}_X\ad
\end{equation}
where the $\ket{\phi_j}$ are non-orthogonal pure states, $\sg_B\in\DC_B$, $\{X_j\}\in \XC_B$, and $\tilde{V}_X = \sum_j \ket{\phi_j} \ot X_j$ is an isometry. States of the form of \eqref{eqn17} can be viewed as resulting from a sort of premeasurement, but where the conditional states on the measurement device $\{\ket{\phi_j}\}$, associated with the different $X_j$ projectors on the system being measured, are not necessarily orthogonal. Hence these states are obtained from doing a ``weak" or ``soft" premeasurement (i.e., not fully extracting the $X$ information), and the lower-case $\textsf{m}$ in $\textsf{mQ}$ emphasizes this.

\section{Collapse of quantum correlation hierarchy}\label{sct5}

\subsection{Four types of quantum correlation measures} \label{sct5.0}

To what degree is the correlation between two systems different from that of a classical joint probability distribution - this is the basic question one aims to answer with quantum correlation measures. This difference can be quantified in a wide variety of ways, but let us consider four common paradigms. (This introduction is for completeness only, please see \cite{HHHH09, ModiEtAl2011review} for review articles.)

One can quantify how far the quantum state is from the set of classical states, $\textsf{CC}$, either in terms of a distance or in terms of the information content of the states. These are sometimes called two-way quantumness or two-way discord measures, since they measure non-classicality with respect to both subsystems.

A second paradigm is to quantify how far the state is from either $\textsf{CQ}$ or $\textsf{QC}$, these are one-way quantumness (or discord) measures, since they measure the non-classicality with respect to just one subsystem. 

A third paradigm is to quantify the distance to $\textsf{Sep}$, these are called entanglement measures. In practice, the label ``entanglement measure" is restricted to those measures that are non-increasing under local operations and classical communication (LOCC) \cite{HHHH09}, though there is some connection between this criterion and quantifying the distance to $\textsf{Sep}$ \cite{VedrPlen1998}.

Finally there are measures of the form of the negative of a conditional entropy, which quantify the distance to a state of the form $\tau_A \ot \rho_B$ where $\tau_A$ is the maximally mixed state (see below), and in the case of von Neumann entropy, the measure is called coherent information.

\subsection{Basic structure of results} \label{sct5.1}

Within each of these four paradigms, there are different quantitative measures, and below we discuss measures based on relative entropies, measures related to the von Neumann entropy, and measures related to smooth entropies. But in each case there is a basic structure that negative conditional entropy (coherent information) lower bounds entanglement which lower bounds one-way discord which lower bounds two-way discord. We call this the \textit{quantum correlation hierarchy}, e.g., Ref.~\cite{PianiAdessoPRA.85.040301} discussed this idea. Many of the inequalities in these hierarchies are well-known, although some require proof. 

In what follows, we present our main technical results, that premeasurement states \textit{collapse the quantum correlation hierarchy}. For these states, which we also call $\textsf{MQ}$ states, defined by \eqref{eqn6} or \eqref{eqn14}, the inequalities relating coherent information, entanglement, one-way discord, and two-way discord turn into \textit{equalities}.

Geometrically speaking, the collapse is some reflection of the fact that the closest separable state to a premeasurement state is a $\textsf{CC}$ state. One can verify this claim (Appendix~\ref{app5}) using the Bures distance \cite{BenZyc06}, a true metric, though in what follows we observe this phenomenon using various relative entropies as (pseudo) measures of distance.

\subsection{Collapse of measures based on relative entropy} \label{sct5.2}

Here we use the relative entropy to express various correlation measures as a distance to a certain class of states \cite{ModiEtAl2010}. In particular, we consider a generalized relative entropy $D_K(P||Q)$, a function that maps two positive-semidefinite operators $P$ and $Q$ to the real numbers, that satisfies the following two properties (also considered in \cite{ColesColbeckYuZwolak2012PRL}):

\begin{enumerate}[label=(\alph{enumi})]
\item \label{a} Non-increasing under quantum channels $\EC$: $D_K(\EC(P)||\EC(Q))\leq D_K(P||Q)$.
\item \label{b} Being unaffected by null subspaces: $D_{K}(P \oplus 0 || Q\oplus Q')=D_{K}(P||Q)$, where $\oplus$ denotes direct sum.  
\end{enumerate}

These properties are satisfied by several important examples \cite{ColesColbeckYuZwolak2012PRL}, and so there is power in formulating a general result that relies only on the properties. Examples include the von Neumann relative entropy \cite{VedralReview02, NieChu00}:
\begin{equation}
\label{eqn18}
D(P || Q) := \Tr(P \log P)- \Tr(P \log Q),
\end{equation}
the Renyi relative entropies \cite{Renyi, Petz84} within the range $\al \in (0,2]$:
\begin{equation}
\label{eqn19}
D_{\al}(P || Q) := \frac{1}{\al-1}\log \Tr (P ^{\al} Q^{1-\al}),
\end{equation}
and the relative entropies associated with the min- and max-entropies \cite{RennerThesis05, KonRenSch09}, respectively,
\begin{align}
\label{eqn20}D_{\max}(P || Q)&:=\log\min\{\lm : P\leq \lm Q \},\\
\label{eqn21}D_{\fid}(P || Q)&:= -2\log \Tr [(\sqrt{P} Q\sqrt{P})^{1/2}].
\end{align}
We label \eqref{eqn20} as $D_{\max}$ (even though it is associated with the min-entropy) because in general $D_{\max}(P  || Q) \geq D( P || Q)$ \cite{Datta09}, and we label \eqref{eqn21} as $D_{\fid}$ because it is closely related to the fidelity.

Consider an entanglement measure \cite{VedrPlen1998} based on $D_K$:
\begin{equation}
\label{eqn22}
\Ebb_K^{A|B}(\rho_{AB}):=\min_{\sg_{AB}\in \textsf{Sep}}D_K(\rho_{AB}|| \sg_{AB}).
\end{equation}
Property~\ref{a} implies, for any LOCC $\Lm$,
\begin{equation}
\label{eqn23}
\Ebb_K^{A|B}(\rho_{AB})\geq \Ebb_K^{A|B}(\Lm(\rho_{AB})),
\end{equation}
which is the well-known monotonicity property \cite{HHHH09}. Let us also define one-way and two-way measures of quantumness (a.k.a.\ discord) \cite{ModiEtAl2011review}:
\begin{align}
\label{eqn24}\Dl_K^{\overrightarrow{A|B}}(\rho_{AB}):=\min_{\sg_{AB}\in \textsf{CQ}}D_K(\rho_{AB}|| \sg_{AB}),\\
\label{eqn25}\Dl_K^{\overleftrightarrow{A|B}}(\rho_{AB}):=\min_{\sg_{AB}\in \textsf{CC}}D_K(\rho_{AB}|| \sg_{AB}).
\end{align}
Finally, let us define a conditional entropy \cite{ColesColbeckYuZwolak2012PRL},
\begin{equation}
\label{eqn26}
H_K(A|B):=\max_{\sigma_B}[-D_K(\rho_{AB}||\id\ot\sigma_B)],
\end{equation}
where the maximization is over all (normalized) density operators $\sigma_B$ on $B$.

To prove our result, we note two additional properties, which were discussed in \cite{ColesColbeckYuZwolak2012PRL}. If $D_K$ satisfies~\ref{a} and~\ref{b}, and if $\tilde{Q}\geq Q$, then
\begin{equation}
\label{eqn27}
D_{K}(P||Q)\geq D_{K}(P||\tilde{Q}),
\end{equation}
and if $\Pi_P$ is a projector onto a space that includes the support of $P$, then
\begin{equation}
\label{eqn28}
D_{K}( P || Q )\geq D_{K}(P|| \Pi_P Q\Pi_P).
\end{equation}

We now show that the correlation measures defined above form a hierarchy. This hierarchy, though interesting in itself, will be useful below in proving that the various correlation measures become equal in the special case of premeasurement states.
\begin{lemma}
\label{thm1}
Let $D_K$ satisfy~\ref{a} and~\ref{b}, then for any $\rho_{AB}$,
\begin{equation} 
\label{eqn29}
-H_K(A|B) \leq \Ebb_K^{A| B}\leq  \Dl_K^{\overrightarrow{A|B}}, \Dl_K^{\overrightarrow{B|A}}\leq \Dl_K^{\overleftrightarrow{A|B}}.
\end{equation} 
\end{lemma}
\begin{proof}
The left-most inequality is proven by supposing $\sg_{AB}\in\textsf{Sep}$ achieves the minimization in $\Ebb_K^{A| B}(\rho_{AB})$, then
\begin{align}
\Ebb_K^{A| B}(\rho_{AB})&=D_K(\rho_{AB}|| \sg_{AB})\notag\\
&\geq D_K(\rho_{AB}|| \id \ot \sg_B)\geq - H_{K}(A|B)\notag
\end{align}
where we invoked~\eqref{eqn27} and the fact that, if $\sg_{AB}$ is separable, then $\id \ot \sg_B\geq \sg_{AB}$ with $\sg_B=\Tr_A(\sg_{AB})$. The other inequalities follow from $\textsf{CC} \subset \textsf{CQ} \subset \textsf{Sep}$.
\end{proof}

Now we can state one of our main technical results, that the hierarchy in \eqref{eqn29} collapses onto a single value for $\textsf{MQ}$ states, which we also call premeasurement states (see Sec.~\ref{sct22b}) since system $A$ plays the role of a measurement device $M_X$ and $B$ is the system $S$ being measured.
\begin{theorem}
\label{thm2}
Let $D_K$ satisfy~\ref{a} and~\ref{b}, then for any premeasurement state $\tilde{\rho}_{M_X S}= V_X\rho_S V_X\ad$,
\begin{align}
\label{eqn30}
-H_K(M_X|S)= \Ebb_K^{M_X| S}=  \Dl_K^{\overrightarrow{M_X|S}}=\Dl_K^{\overrightarrow{S|M_X}}= \Dl_K^{\overleftrightarrow{M_X|S}}.
\end{align}
\end{theorem}
\begin{proof}
Let $\sg_S$ be the state that achieves the optimization in $H_{K}(M_X|S)$, then
\begin{align}
-H_K(M_X|S) &=D_K(\tilde{\rho}_{M_X S}|| \id \ot \sg_S ) \notag\\
&\geq D_K(\tilde{\rho}_{M_X S}|| V_XV_X\ad (\id \ot \sg_S) V_XV_X\ad ) \notag\\
&=D_K(\tilde{\rho}_{M_X S}|| \sum_j \dya{j}\ot X_j\sg_SX_j)\notag\\
&\geq \Dl_K^{\overleftrightarrow{M_X|S}}. \notag
\end{align}
We used \eqref{eqn28} in the second line, and the last inequality notes that $\sum_j \dya{j}\ot X_j\sg_SX_j \in \textsf{CC}$. (Expand the $X_j\sg_SX_j $ blocks in their eigenbasis to verify this.) But \eqref{eqn29} gives an inequality in the reverse direction, so the entire hierarchy in \eqref{eqn29} must collapse onto the same value.
\end{proof}

Theorem~\ref{thm2} applies to all the relative entropies listed in \eqref{eqn18} through \eqref{eqn21}. For example, in the case of von Neumann relative entropy, the quantities in \eqref{eqn30}, from left to right, are the coherent information \cite{NieChu00}, the relative entropy of entanglement \cite{VedrPlen1998}, the one-way information deficit \cite{HorEtAl05}, and the relative entropy of quantumness \cite{HorEtAl05}. In the next subsection, we further extend our results for these von Neumann measures, including other measures into the hierarchy collapse.

\subsection{Collapse of von Neumann measures}\label{sct5.3}

\subsubsection{Long list of measures involved in the collapse}

Here we elaborate on the hierarchy collapse for von Neumann measures, giving a long list of the measures involved. While operational or conceptual meanings of many of the measures can be found in \cite{HHHH09, ModiEtAl2011review}, this article is more concerned with the fact that they form a hierarchy and that this hierarchy collapses for $\MQ$ states. To illustrate the dramatic effect of the collapse, we attempt to demonstrate it for as many measures as possible here, even though it comes at the expense of having to define many quantities.

In the previous subsection, we considered the coherent information $I_c$ \cite{NieChu00}, relative entropy of entanglement $\Ebb_R$ \cite{VedrPlen1998}, one-way information deficit $\Dl^{\to}$ \cite{HorEtAl05}, and relative entropy of quantumness $\Dl^{\leftrightarrow}$ \cite{HorEtAl05}, respectively defined by:
\begin{align}
\label{eqn31}
I_c^{\overrightarrow {A|B}}(\rho_{AB})&:=-H(A|B)=D(\rho_{AB}|| \id \ot \rho_B),\notag\\
\Ebb_R^{A|B}(\rho_{AB})&:=\min_{\sg_{AB}\in \textsf{Sep}}D(\rho_{AB}|| \sg_{AB}), \notag\\
\Dl^{\overrightarrow{A|B}}(\rho_{AB})&:=\min_{\sg_{AB}\in \textsf{CQ}}D(\rho_{AB}|| \sg_{AB}), \notag\\
\Dl^{\overleftrightarrow{A|B}}(\rho_{AB})&:=\min_{\sg_{AB}\in \textsf{CC}}D(\rho_{AB}|| \sg_{AB}). 
\end{align}
We note that $I_c$ appears in the expression for the quantum capacity of a quantum channel \cite{Lloyd97}, is related to the entanglement distillable through one-way hashing \cite{DevWin05}, and has been interpreted as the entanglement gained in quantum state merging \cite{HorOppWin05}.

We will also consider discord measures \cite{OllZur01, WuEtAl2009} based on a difference of quantum mutual informations $I(\rho)$, defined as follows
\begin{align}
\label{eqn32}
\dl^{\overrightarrow{A|B}}(\rho_{AB}):=\min_{\YC} \{I(\rho_{AB})-I[(\YC \ot \id)(\rho_{AB})]\},\notag\\
\dl^{\overleftrightarrow{A|B}}(\rho_{AB}):= \min_{\YC,\YC'} \{I(\rho_{AB})-I[(\YC \ot \YC')(\rho_{AB})]\}.
\end{align}
Here, we suppose that $\{Y_j\}$ and $\{Y'_k\}$ are positive operator valued measures (POVMs) on $A$ and $B$, respectively, and the quantum channels $\YC$ and $\YC'$ associated with these POVMs are defined such that 
\begin{align}
(\YC \ot \id)(\rho_{AB})&:=\sum_j \dya{j}\ot \Tr_A(Y_j \rho_{AB}), \notag \\
(\YC \ot \YC')(\rho_{AB})&:=\sum_{j,k} \Tr [(Y_j\ot Y'_k) \rho_{AB}] \dya{j}\ot \dya{k},\notag 
\end{align}
with $\{\ket{j}\}$ being the standard (orthonormal) basis.

Now we define regularized versions of these measures:
\begin{align}
\label{eqn34}
I_{c,\infty}^{\overrightarrow {A|B}}(\rho_{AB})&:=\lim_{N\to \infty} (1/N) I_c^{\overrightarrow{A^{\ot N}|B^{\ot N}}}(\rho_{AB}^{\ot N}), \notag\\
\Ebb_{R,\infty}^{A|B}(\rho_{AB})&:= \lim_{N\to \infty}(1/N)E_R^{A^{\ot N}|B^{\ot N}}(\rho_{AB}^{\ot N}), \notag\\
\Dl_{\infty}^{\overrightarrow{A|B}}(\rho_{AB})&:= \lim_{N\to \infty}(1/N) \Dl^{\overrightarrow{A^{\ot N}|B^{\ot N}}}(\rho_{AB}^{\ot N}), \notag\\
\Dl_{\infty}^{\overleftrightarrow{A|B}}(\rho_{AB})&:= \lim_{N\to \infty}(1/N) \Dl^{\overleftrightarrow{A^{\ot N}|B^{\ot N}}}(\rho_{AB}^{\ot N}), \notag\\
\dl_{\infty}^{\overrightarrow{A|B}}(\rho_{AB})&:= \lim_{N\to \infty}(1/N) \dl^{\overrightarrow{A^{\ot N}|B^{\ot N}}}(\rho_{AB}^{\ot N}),\notag\\
\dl_{\infty}^{\overleftrightarrow{A|B}}(\rho_{AB})&:= \lim_{N\to \infty}(1/N) \dl^{\overleftrightarrow{A^{\ot N}|B^{\ot N}}}(\rho_{AB}^{\ot N}) .
\end{align}
From the additivity of the von Neumann relative entropy, we have
$$I_{c,\infty}^{\overrightarrow {A|B}}=I_c^{\overrightarrow {A|B}},$$
and it was shown in \cite{Devetak07} and discussed in \cite{HorEtAl05} that
\begin{equation}
\label{eqn35}
\dl_{\infty}^{\overrightarrow{A|B}}=\Dl_{\infty}^{\overrightarrow{A|B}}.
\end{equation}
In asymptotia, $\Ebb_{R,\infty}$ uniquely characterises the amount of entanglement in a state when all non-entangling transformations are allowed \cite{BranPlen2008}, while $\dl_{\infty}^{\to}$ has been linked to entanglement irreversibility (when dilution and distillation are respectively done by LOCC and hashing) in a tripartite scenario \cite{CorDeOFan11}.

In what follows, we also consider the distillable entanglement $\Ebb_D$ and the distillable secret key $K_D$ \cite{HHHH09}, both of which are asymptotic rates for conversion of many copies of $\rho_{AB}$ into some resource, where the resource is EPR pairs and bits of secret correlation, respectively, for $\Ebb_D$ and $K_D$.

Now we consider some hierarchies satisfied by the above measures. As mentioned, the basic structure for these hierarchies is that coherent information lower bounds entanglement which lower bounds one-way discord which lower bounds two-way discord, and indeed \eqref{eqn36}--\eqref{eqn38} below each have this form. Equation~\eqref{eqn36} involves discord based on relative entropy whereas \eqref{eqn37} involves discord based on a difference of mutual informations, and Eq.~\eqref{eqn38} involve regularised versions of these measures. Thus, individually, each equation in the following Lemma, proved in Appendix~\ref{app4}, can be regarded as a quantum correlation hierarchy. 
\begin{lemma}
\label{thm3}
For any bipartite state $\rho_{AB}$,
\begin{align}
\label{eqn36}&I_c^{\overrightarrow {A|B}}\leq \Ebb_{D}^{A|B}\leq K_{D}^{A|B} \leq \Ebb_{R}^{A|B}\leq \Dl^{\overrightarrow{A|B}}, \Dl^{\overrightarrow{B|A}}\leq \Dl^{\overleftrightarrow{A|B}},\\
\label{eqn37}&I_c^{\overrightarrow {A|B}}\leq \Ebb_{D}^{A|B}\leq K_{D}^{A|B}\leq \dl^{\overrightarrow{A|B}}, \dl^{\overrightarrow{B|A}} \leq \dl^{\overleftrightarrow{A|B}} \leq \Dl^{\overleftrightarrow{A|B}},\\
\label{eqn38}&I_c^{\overrightarrow {A|B}} \leq \Ebb_{R,\infty}^{A|B}  \leq \dl_{\infty}^{\overrightarrow{A|B}}, \dl_{\infty}^{\overrightarrow{B|A}} \leq \dl_{\infty}^{\overleftrightarrow{A|B}} \leq \Dl_{\infty}^{\overleftrightarrow{A|B}} \leq \Dl^{\overleftrightarrow{A|B}}.
\end{align}
\end{lemma}
We now see that each of these hierarchies \eqref{eqn36}--\eqref{eqn38} collapses in the special case where the state is $\MQ$. In fact, the hierarchies themselves are useful in proving the collapse. In Theorem~\ref{thm2}, we showed that, if $\tilde{\rho}_{M_XS}\in \textsf{MQ}$, then
$$I_c^{\overrightarrow{M_X|S}}= \Dl^{\overleftrightarrow{M_X|S}},$$
so combining this with Lemma~\ref{thm3} immediately implies the following result.
\begin{theorem}
\label{thm4}
For any state in $\textsf{MQ}$, i.e., any premeasurement state $\tilde{\rho}_{M_XS}= V_X\rho_S V_X\ad$,
\begin{align}
\label{eqn39}
I_c^{\overrightarrow{M_X|S}}&=  \Ebb_{D}^{M_X|S}=K_{D}^{M_X|S}=\Ebb_{R,\infty}^{M_X|S}=\Ebb_{R}^{M_X|S}\notag\\
&= \dl_{\infty}^{\overrightarrow{M_X|S}}= \dl_{\infty}^{\overrightarrow{S|M_X}}=\dl^{\overrightarrow{M_X|S}}= \dl^{\overrightarrow{S|M_X}} \notag\\
&= \Dl^{\overrightarrow{M_X|S}}= \Dl^{\overrightarrow{S|M_X}}=\dl_{\infty}^{\overleftrightarrow{M_X|S}}=\Dl_{\infty}^{\overleftrightarrow{M_X|S}} \notag\\
&=\dl^{\overleftrightarrow{M_X|S}}=\Dl^{\overleftrightarrow{M_X|S}}.
\end{align}
\end{theorem}

While the list in Theorem~\ref{thm4} is very long, we note that not all measures participate in the collapse for $\textsf{MQ}$. For example, $I_c^{\overrightarrow{S|M_X}}=-H(S|M_X)$ need not be equal to the other correlation measures appearing above. One can see this as follows. For the state $\tilde{\rho}_{M_XS}$ which is purified by $E$ to the state $\tilde{\rho}_{M_XSE}$, we have:
\begin{equation}
H(S|M_X)-H(M_X|S)=H(E|M_X)=\sum_j p_j H(\rho_{E,j}),\notag
\end{equation}
where $\tilde{\rho}_{M_XE}=\sum_j p_j \dya{j}\ot \rho_{E,j}\in \textsf{CQ}$. Hence if the $\rho_{E,j}$ are non-pure, then $[-H(S|M_X)]$ will not collapse onto the other measures. Also, see \cite{CorDeOFan11} for a discussion of entanglement of formation and entanglement cost.

\subsubsection{Is the collapse unique to $\textsf{MQ}$?}

Here we give a simple argument that $\textsf{MQ} $ is the \textit{only} set of bipartite states for which $I_c^{\overrightarrow{A|B}}= \Dl^{\overleftrightarrow{A|B}} $, and hence the only set that collapses the \textit{full} hierarchy as in Theorem~\ref{thm4}. Let $C$ purify $\rho_{AB}$, then it is straightforward to show that:
\begin{align}
I_c^{\overrightarrow{A|B}}= \dl^{\overrightarrow{A|B}}-\dl^{\overrightarrow{A|C}},
\end{align}
by noting that the optimization in $\dl^{\to}$ is achieved by a rank-one POVM, and in fact the same rank-one POVM achieves the optimization in both $\dl^{\overrightarrow{A|B}}$ and $\dl^{\overrightarrow{A|C}}$ \cite{ColesEtAl}.  From Ref.~\cite{DattaArxiv2010} and the definition of $\MQ$, we have:
$$\dl^{\overrightarrow{A|C}}=0 \Leftrightarrow \rho_{AC}\in \textsf{CQ} \Leftrightarrow \rho_{AB}\in \textsf{MQ} .$$
Therefore, for any $\rho_{AB}$ that is \textit{not} in $\textsf{MQ}$, we have $\dl^{\overrightarrow{A|C}}>0$ (for all purifications $\rho_{ABC}$ of $\rho_{AB}$) and 
\begin{align}
\label{eqn40}
I_c^{\overrightarrow{A|B}} < \dl^{\overrightarrow{A|B}} \leq \Dl^{\overleftrightarrow{A|B}},
\end{align}
showing that $\textsf{MQ}$ is the only set of states for which $I_c^{\overrightarrow{A|B}} = \dl^{\overrightarrow{A|B}}$, and hence the only set for which $I_c^{\overrightarrow{A|B}}=\Dl^{\overleftrightarrow{A|B}}$.

We wish to emphasize that other states besides $\textsf{MQ}$ states may collapse ``part" of the hierarchy. For example, consider a tensor product of maximally mixed states, say, of the form $\rho_{AB}= ( \id / d) \ot ( \id / d)$. Clearly all measures of entanglement and discord are zero for this state. But the coherent information is $I_c^{\overrightarrow{A|B}} = - \log d$, and this state is not an $\textsf{MQ}$ state.
 
Likewise, as mentioned in Section~\ref{sct22b}, a superset of $\MQ$, denoted $\mQ$, partially collapses the hierarchy, as shown in \cite{CorDeOFan11}. Specifically, Ref.~\cite{CorDeOFan11} showed that
\begin{align}
\label{eqn41}
I_c^{\overrightarrow{A|B}}= \Ebb_{D}^{A|B}=K_{D}^{A|B}=\Ebb_{R,\infty}^{A|B}= \dl^{\overrightarrow{B|A}}_{\infty} =\dl^{\overrightarrow{B|A}}
\end{align}
for $\rho_{AB}\in \textsf{mQ}$. However, \eqref{eqn40} indicates that, for those states in $\mQ $ that are not in $\MQ $, there is a gap between the ``collapsed measures" appearing in \eqref{eqn41} and a particular one-way discord, $\dl^{\overrightarrow{A|B}}$.

\subsection{Collapse of smooth measures}\label{sct5.4}

While there are various correlation hierarchies that we could investigate, we have been focusing on those that involve a conditional entropy as one of the measures. This is because we will ultimately be interested in using the hierarchy collapse to reinterpret entropic uncertainty relations (EURs), which are often formulated using conditional entropies. One such EUR has been formulated for smooth entropies \cite{TomRen2010}, and so we will consider the correlation hierarchy related to smooth entropies in this subsection, again with the intention of giving a reinterpretation of this EUR. 

Smooth entropies pose a dilemma in that they are highly powerful tools relevant to non-asymptotic information-processing tasks \cite{RennerThesis05URL, TomamichelThesis2012}, yet they are quite technical. We therefore give only the main results in this section, and relegate all proofs to (a lengthy) Appendix~\ref{app44}.

We start with the min- and max-entropies \cite{KonRenSch09}, 
\begin{align}
&H_{\min}(A|B)_{\rho}:= \max_{\sg_B}[-D_{\max}(\rho_{AB}|| \id \ot \sg_B)], \notag\\
&H_{\max}(A|B)_{\rho}:= \max_{\sg_B}[-D_{\fid}(\rho_{AB}|| \id \ot \sg_B)], \notag
\end{align}
where the maximization is over all normalized density operators $\sg_B$, and $D_{\max}$ and $D_{\fid}$ were defined in \eqref{eqn20} and \eqref{eqn21}.

To define the smooth entropy of $\rho$, we optimize the entropy over a ball of radius $\ep$ centered around $\rho$ in the space of subnormalized positive operators, denoted $\BC^{\ep}(\rho)$. We use the purified distance to define this ball \cite{TomColRen10}, again with all the details in Appendix~\ref{app44}. Then, the smooth min- and max-entropies are defined as \cite{TomamichelThesis2012, TomColRen10}:
\begin{align}
&H^{\ep}_{\min}(A|B)_{\rho}:=\max_{\sg_{AB} \in \BC^{\ep}(\rho)} H_{\min}(A|B)_{\sg}, \notag\\
&H^{\ep}_{\max}(A|B)_{\rho}:= \min_{\sg_{AB} \in \BC^{\ep}(\rho)} H_{\max}(A|B)_{\sg}. \notag
\end{align}
Note that a maximisation (minimisation) is performed for the smooth min (max) entropy; in this form these are the relevant quantities for characterising the operational tasks involved in quantum key distribution \cite{RennerThesis05URL, TomamichelThesis2012}. 

To obtain results that are mathematically analogous to Lemma~\ref{thm1} and Theorem~\ref{thm2}, we will need to define smooth measures of entanglement and discord. We note that smooth measures of entanglement were considered, e.g., in \cite{BusDatPRL2011, BranDatt2011}. Consider first the unsmooth measures of entanglement and discord (one-way and two-way) based on the max relative entropy, respectively given by:
\begin{align}
\Ebb_{\max}^{A|B}(\rho_{AB}):=\min_{\sg_{AB}\in \textsf{Sep}}D_{\max}(\rho_{AB}|| \sg_{AB}),\notag\\
\Dl_{\max}^{\overrightarrow{A|B}}(\rho_{AB}):=\min_{\sg_{AB}\in \textsf{CQ}}D_{\max}(\rho_{AB}|| \sg_{AB}),\notag\\
\Dl_{\max}^{\overleftrightarrow{A|B}}(\rho_{AB}):=\min_{\sg_{AB}\in \textsf{CC}}D_{\max}(\rho_{AB}|| \sg_{AB}).\notag
\end{align}
and consider analogous quantities $\Ebb_{\fid}^{A|B}$, $\Dl_{\fid}^{\overrightarrow{A|B}}$, and $\Dl_{\fid}^{\overleftrightarrow{A|B}}$ defined similarly but with $D_{\max}$ replaced by $D_{\fid}$. We note that $\Ebb_{\fid}$ and $\Ebb_{\max}$ are non-increasing under LOCC due to Property~\ref{a}, and $\Ebb_{\max}$ was characterised in \cite{Datta09}.

We now define smooth versions of these quantum correlation measures, as follows:
\begin{align}
\label{eqn42}
_{\ep}\Ebb_{\max}^{A|B}(\rho_{AB}):= \min_{\sg_{AB} \in \BC^{\ep}(\rho)} \Ebb_{\max}^{A|B}(\sg_{AB}),\notag\\
_{\ep}\Dl_{\max}^{\overrightarrow{A|B}}(\rho_{AB}):= \min_{\sg_{AB} \in \BC^{\ep}(\rho)} \Dl_{\max}^{\overrightarrow{A|B}}(\sg_{AB}),\notag\\
_{\ep}\Dl_{\max}^{\overleftrightarrow{A|B}}(\rho_{AB}):= \min_{\sg_{AB} \in \BC^{\ep}(\rho)} \Dl_{\max}^{\overleftrightarrow{A|B}}(\sg_{AB}),
\end{align}
and 
\begin{align}
\label{eqn43}
_{\ep}\Ebb_{\fid}^{A|B}(\rho_{AB}):= \max_{\sg_{AB} \in \BC^{\ep}(\rho)} \Ebb_{\fid}^{A|B}(\sg_{AB}),\notag\\
_{\ep}\Dl_{\fid}^{\overrightarrow{A|B}}(\rho_{AB}):= \max_{\sg_{AB} \in \BC^{\ep}(\rho)} \Dl_{\fid}^{\overrightarrow{A|B}}(\sg_{AB}),\notag\\
_{\ep}\Dl_{\fid}^{\overleftrightarrow{A|B}}(\rho_{AB}):= \max_{\sg_{AB} \in \BC^{\ep}(\rho)} \Dl_{\fid}^{\overleftrightarrow{A|B}}(\sg_{AB}).
\end{align}
A smooth max entanglement defined similarly to the one in \eqref{eqn42} was previously given an operational meaning in terms of one-shot catalytic entanglement cost under non-entangling maps \cite{BranDatt2011}. We note that performing a minimisation in \eqref{eqn42} and a maximisation in \eqref{eqn43} appears to be necessary to obtain the generalisation of our results to smooth measures. 

We now state an analog of Lemma~\ref{thm1} for smooth measures, where Eqs.~\eqref{eqn44} and \eqref{eqn45} below can be viewed as quantum correlation hierarchies involving the smooth min and max entropies, respectively.
\begin{lemma}\label{thm5}
For any bipartite state $\rho_{AB}$,
\begin{align}
\label{eqn44}&-H^{\ep}_{\min}(A|B) \leq  {}_{\ep}\Ebb^{A|B}_{\max} \leq  {}_{\ep}\Dl^{\overrightarrow{A|B}}_{\max}, {}_{\ep}\Dl^{\overrightarrow{B|A}}_{\max} \leq  {}_{\ep}\Dl^{\overleftrightarrow{A|B}}_{\max}, \\
\label{eqn45}&-H^{\ep}_{\max}(A|B) \leq  {}_{\ep}\Ebb^{A|B}_{\fid} \leq  {}_{\ep}\Dl^{\overrightarrow{A|B}}_{\fid},  {}_{\ep}\Dl^{\overrightarrow{B|A}}_{\fid}  \leq  {}_{\ep}\Dl^{\overleftrightarrow{A|B}}_{\fid}.
\end{align}
\end{lemma}

Analogous to Theorem~\ref{thm2}, we find that the hierarchies of smooth quantum correlation measures in \eqref{eqn44} and \eqref{eqn45} collapse in the special case of premeasurement states.
\begin{theorem}\label{thm6}
For any state in $\textsf{MQ}$, i.e., any premeasurement state $\tilde{\rho}_{M_XS}= V_X\rho_S V_X\ad$,
\begin{align}
\label{eqn46} -H^{\ep}_{\min}(M_X|S) &=  {}_{\ep}\Ebb^{M_X|S}_{\max} =  {}_{\ep}\Dl^{\overrightarrow{M_X|S}}_{\max}\notag\\
&={}_{\ep}\Dl^{\overrightarrow{S|M_X}}_{\max} =  {}_{\ep}\Dl^{\overleftrightarrow{M_X|S}}_{\max},\\
\label{eqn47} -H^{\ep}_{\max}(M_X|S) & =  {}_{\ep}\Ebb^{M_X|S}_{\fid} =  {}_{\ep}\Dl^{\overrightarrow{M_X|S}}_{\fid}\notag\\
&={}_{\ep}\Dl^{\overrightarrow{S|M_X}}_{\fid} =  {}_{\ep}\Dl^{\overleftrightarrow{M_X|S}}_{\fid}.
\end{align}
\end{theorem}
We note that these smooth measures reduce to the corresponding non-smooth measures for $\ep = 0$. Hence, we had already proved Lemma~\ref{thm5} and Theorem~\ref{thm6} for the special case of $\ep = 0$ in Section~\ref{sct5.2}, but the smooth versions of these results, valid for any $\ep \geq 0$, are a significant generalization. While superficially it seems simple to add an $\ep$ as a superscript or subscript, let the reader beware that the proof of this result for smooth measures is non-trivial.

\section{Connection to uncertainty}\label{sct6}

We have investigated several quantum correlation hierarchies, and in each case we found that premeasurement states collapse the hierarchy. We would now like to take advantage of the dynamic view, shown schematically in Fig.~\ref{fgr1}, that these states are produced during the measurement process. In principle, premeasurement states can range from being maximally entangled to being only classically correlated to being completely uncorrelated. What features of the state \textit{prior} to the controlled-shift operation in Fig.~\ref{fgr1} determine the correlations of the premeasurement state? As we will see, it is the \textit{uncertainty} of the observable being measured that ultimately determines the correlations produced during the premeasurement.

The key property that allows us to connect uncertainty to the quantum correlations of premeasurement states is the tripartite duality of conditional entropy functions. For example, for the von Neumann entropy, we have:
\begin{equation}
\label{eqn48}
H(A|B) = - H(A|C)
\end{equation}
for any pure state on $\HC_{ABC}$. Let us apply this duality to the pure state $\rhot_{M_XSE} = V_X\dya{\psi}V_X\ad$ shown in Fig.~\ref{fgr1}, giving:
\begin{equation}
\label{eqn49}
H(M_X|E)_{\rhot} = - H(M_X|S)_{\rhot}
\end{equation}
Now we note that the left side of \eqref{eqn49} is the standard way of defining the uncertainty of an observable conditioned on quantum memory \cite{RenesBoileau, BertaEtAl, TomRen2010, ColesEtAl, ColesColbeckYuZwolak2012PRL}. That is, $H(X|E)_{\rho} := H(M_X|E)_{\rhot}$, the uncertainty of observable $X$ when the observer is given access to system $E$ is defined as the quantum conditional entropy of $M_X$ given $E$ at the end of the process depicted in Fig.~\ref{fgr1}. In addition, Theorem~\ref{thm4} showed that the right side of \eqref{eqn49} is equal to a long list of other quantum correlation measures, so we have:
\begin{align}
\label{eqn50}
H(X|E)&=  \Ebb_{D}^{M_X|S}=K_{D}^{M_X|S}=\Ebb_{R,\infty}^{M_X|S}=\Ebb_{R}^{M_X|S}\notag\\
&= \dl_{\infty}^{\overrightarrow{M_X|S}}= \dl_{\infty}^{\overrightarrow{S|M_X}}=\dl^{\overrightarrow{M_X|S}}= \dl^{\overrightarrow{S|M_X}} \notag\\
&= \Dl^{\overrightarrow{M_X|S}}= \Dl^{\overrightarrow{S|M_X}} =\dl_{\infty}^{\overleftrightarrow{M_X|S}}= \Dl_{\infty}^{\overleftrightarrow{M_X|S}}\notag\\
&=\dl^{\overleftrightarrow{M_X|S}}=\Dl^{\overleftrightarrow{M_X|S}},
\end{align}
where it should be understood that the measures on the right side are applied to the state $\rhot_{M_XS}$. We note that a preliminary version of \eqref{eqn50} appeared in Theorem 2 of Ref.~\cite{ColesDecDisc2012}, but our results here go significantly beyond the related results in \cite{ColesDecDisc2012}.

This is a fascinating connection. It says that the uncertainty of an observable, given the environment, is a measure of quantum correlations (entanglement, discord, etc.) produced when that observable is measured. When the system's initial state is pure, the environment $E$ can be ignored and the left side of \eqref{eqn50} becomes $H(X)$, the Shannon entropy of the $X$ observable.

We remark that we have assumed $X$ is a projective (but not necessarily fine-grained) observable, i.e., a PVM. To generalize \eqref{eqn50} to the case where $X$ is a POVM, simply replace $S$ on the right side with a Naimark extension $\Sbf$, i.e., an enlargement of the system's Hilbert space that allows $X$ to be thought of as a projective observable. Such an extension can be found for any POVM.

Now let us consider the analog of \eqref{eqn50} for other entropies. Consider the generic conditional entropy $H_K(A|B)$ introduced in Sect.~\ref{sct5.2}, based on a generic relative entropy $D_K$. In \cite{ColesColbeckYuZwolak2012PRL}, it was shown that, because of Properties~\ref{a} and \ref{b}, $H_K$ is guaranteed to have a \emph{dual} entropy $H_{\Khat}$ that is well-defined by
\begin{equation}
\label{eqn51}
H_{\Khat}(A|B):=-H_K(A|C),
\end{equation}
where $\rho_{ABC}$ is a purification of $\rho_{AB}$. Again let us apply this duality to $\rhot_{M_XSE}$ to obtain $H_{\Khat}(M_X|E)_{\rhot}=-H_K(M_X|S)_{\rhot}$, invoke the standard definition for uncertainty with quantum memory $H_{\Khat}(X|E)_{\rho} := H_{\Khat}(M_X|E)_{\rhot}$, and combine this with Theorem~\ref{thm2} to find:  
\begin{align}
\label{eqn52}
H_{\Khat}(X|E)= \Ebb_K^{M_X| S}=  \Dl_K^{\overrightarrow{M_X|S}}=\Dl_K^{\overrightarrow{S|M_X}}= \Dl_K^{\overleftrightarrow{M_X|S}}.
\end{align}
This gives a fairly general connection between an observable's uncertainty and the quantum correlations created upon its measurement, e.g., applicable when the correlation measures are based on any of the relative entropies in \eqref{eqn18}--\eqref{eqn21}. For example, because the min- and max-entropies are dual to each other \cite{KonRenSch09}, \eqref{eqn52} implies that
\begin{align}
\label{eqn53}
H_{\min}(X|E)&= \Ebb_{\fid}^{M_X| S}=  \Dl_{\fid}^{\overrightarrow{M_X|S}}=\Dl_{\fid}^{\overrightarrow{S|M_X}}= \Dl_{\fid}^{\overleftrightarrow{M_X|S}},\notag\\
H_{\max}(X|E)&= \Ebb_{\max}^{M_X| S}=  \Dl_{\max}^{\overrightarrow{M_X|S}}=\Dl_{\max}^{\overrightarrow{S|M_X}}= \Dl_{\max}^{\overleftrightarrow{M_X|S}}.
\end{align}

Now consider the \textit{smooth} min- and max- entropies discussed in Section~\ref{sct5.4}. They are dual to each other \cite{TomColRen10} in that
$$H^{\ep}_{\max}(A|B) = -H^{\ep}_{\min}(A|C),$$
for pure $\rho_{ABC}$. Again applying this duality to $\rhot_{M_XSE}$ and combining with Theorem~\ref{thm6} gives
\begin{align}
\label{eqn54}
H^{\ep}_{\min}(X|E) =  {}_{\ep}\Ebb^{M_X|S}_{\fid} =  {}_{\ep}\Dl^{\overrightarrow{M_X|S}}_{\fid}={}_{\ep}\Dl^{\overrightarrow{S|M_X}}_{\fid} =  {}_{\ep}\Dl^{\overleftrightarrow{M_X|S}}_{\fid},\notag\\
H^{\ep}_{\max}(X|E) =  {}_{\ep}\Ebb^{M_X|S}_{\max} =  {}_{\ep}\Dl^{\overrightarrow{M_X|S}}_{\max}={}_{\ep}\Dl^{\overrightarrow{S|M_X}}_{\max} =  {}_{\ep}\Dl^{\overleftrightarrow{M_X|S}}_{\max}.
\end{align}
This reduces to \eqref{eqn53} in the case $\ep = 0$, and hence generalizes the connection between uncertainty and the creation of quantum correlations to any $\ep \geq 0$. It is worth remarking that the smooth entropies on the left side of \eqref{eqn54} have important operational meanings in one-shot randomness extraction and data compression \cite{TRSS10, RenesRenner2012, RennerThesis05URL}, which is typically the motivation for their study. While we have not yet established operational meanings for the quantities on the right side of \eqref{eqn54} (though a similar smooth max-entanglement was given an operational meaning in \cite{BranDatt2011}), the connection nonetheless seems interesting, one reason being the validity for any value of $\ep$, suggesting that there truly is a deep connection between uncertainty and the quantum correlations produced in measurements.

\section{Reinterpreting Entropic Uncertainty Relations}\label{sct7}

\subsection{Introduction}\label{sct7.1}

The uncertainty principle plays a crucial role in our understanding of quantum mechanics, expressing a fundamental limit on our knowledge of certain pairs of observables. This idea, with no classical analog, has been captured quantitatively by so-called uncertainty relations, which in modern times typically have the form of a lower bound on the sum of the entropies of different observables and hence called entropic uncertainty relations (EURs). Though the field dates back to Heisenberg, research on the uncertainty principle has seen a sort of revolution in recent years as it was realized \cite{RenesBoileau, BertaEtAl} that the observer can possess quantum memory (a quantum system that could be entangled with the system of interest), and hence we should try to formulate the uncertainty principle within this more general context. This, along with the rise of quantum information theory, has led to a wide variety of EURs \cite{EURreview1} expressed using various entropy functions, some of which allow for quantum memory \cite{RenesBoileau, BertaEtAl, TomRen2010, ColesEtAl, ColesColbeckYuZwolak2012PRL}. 

The results of this article imply that there exists an interpretation of these EURs that is quite different from the typical one as constraints on our knowledge. The uncertainties appearing in these EURs have the form, for example, of the left-hand-sides of \eqref{eqn50}, \eqref{eqn52}, and \eqref{eqn54}. But we have shown that the uncertainty of an observable is quantitatively connected to, e.g., the entanglement created when that observable is measured. Hence, EURs have an interpretation that has nothing to do with uncertainty: they are lower bounds on the entanglement created when incompatible observables are measured! We illustrate this alternative view with a game, in what follows.

\subsection{Entanglement distillation game} \label{sct7.2}

Here we focus on \eqref{eqn50}, in particular, the portion that reads:
\begin{equation}
\label{eqn55}
H(X|E)=  \Ebb_{D}^{M_X|S}
\end{equation}
where $\Ebb_{D}$ is the distillable entanglement \cite{HHHH09}, i.e., the optimal rate to distill EPR pairs using LOCC in the asymptotic limit (infinitely many copies of the state). Again, note that when the initial state of the system, $\rho_S$ in Fig.~\ref{fgr1}, is pure then \eqref{eqn55} becomes $H(X)=  \Ebb_{D}^{M_X|S}$.

Equation~\eqref{eqn55} gives an operational meaning to uncertainty relations written in terms of Shannon entropies \cite{EURreview1}, or ``Shannon uncertainty relations". We illustrate this with the following game, where Alice wants to establish entanglement with Bob but Eve (the adversary) wants to prevent this. Suppose the game is set up such that Eve feeds Alice an (unknown to Alice) pure state $\ket{\psi}_S$ of a qubit $S$ and a register qubit $M$ known to be in the $\ket{0}$ state. Alice is allowed to perform a CNOT between $S$ and $M$, such that some basis on $S$ controls the NOT on the $M$ qubit, and then she can send the $M$ qubit (over a perfect quantum channel) to Bob.
The only freedom Alice is allowed is to change the basis that controls the CNOT. Suppose they repeat this $3N$ times, where $N$ is very large ($N\to \infty$), and each time Eve feeds Alice the same states. At the end of the game, Eve announces the state $\ket{\psi}_S$, and Alice's and Bob's task is now to distill at least $2N$ EPR pairs from their $3N$ pairs of qubits using LOCC, \textit{no matter what $\ket{\psi}_S$ was}. A winning strategy is for Alice to use the $X$, $Y$, and $Z$ bases (three mutually orthogonal axes of the Bloch sphere) each $N$ times. Then, since $\Ebb_D $ is additive here (see Appendix~\ref{app1}), the number of distillable EPR pairs is $N(\Ebb_D^{M_X|S}+\Ebb_D^{M_Y|S}+\Ebb_D^{M_Z|S})$. From \eqref{eqn55} and an uncertainty relation from \cite{SanchezRuiz1995} we have:
$$\Ebb_D^{M_X|S}+\Ebb_D^{M_Y|S}+\Ebb_D^{M_Z|S} =H(X)+H(Y)+H(Z)\geq 2.$$
So, regardless of $\ket{\psi}_S$, the number of distillable EPR pairs is lower-bounded by $2N$. (This also gives an operational meaning to minimum uncertainty states of Shannon uncertainty relations \cite{ColesYuZwo2011}; in this example they have an EPR yield of precisely $2N$.) However, Eve can beat this protocol by feeding Alice a \textit{mixed} state $\rho_S$ and keeping the purifying system $E$. On the other hand, Alice can partially salvage the situation if she can somehow get a hold of a subsystem $E_1$ of $E=E_1E_2$ such that $SE_2$ is in a separable state. In this case, the uncertainty principle with quantum memory \cite{BertaEtAl}, combined with \eqref{eqn55}, gives:
\begin{align}
&\Ebb_D^{M_X|SE_1}+\Ebb_D^{M_Y|SE_1}+\Ebb_D^{M_Z|SE_1} =\notag\\
& H(X|E_2) + H(Y|E_2)+H(Z|E_2) \geq \frac{3}{2}[1+H(S|E_2)]\geq \frac{3}{2},\notag
\end{align}
since the von Neumann conditional entropy $H(S|E_2)\geq 0$ for separable states \cite{NieChu00}. So at least, in this case, Alice and Bob are assured to get $(3/2)N$ EPR pairs. This game illustrates that Shannon uncertainty relations are useful for designing protocols to create entanglement, whenever the state of one's system is unknown.

\subsection{Entanglement creation view of other EURs} \label{sct7.3}

We discussed EURs written in terms of Shannon entropies in the previous subsection, but EURs have been found for other entropies as well. Of particular interest are the min- and max-entropies because they have operational meanings \cite{KonRenSch09}, and more generally, the smooth min- and max-entropies have operational meanings \cite{TRSS10, RenesRenner2012, RennerThesis05URL}.

Let us consider an EUR proved by Tomamichel and Renner for the min- and max-entropies \cite{TomRen2010}. Consider any two POVMs $X=\{X_j\}$ and $Z=\{Z_k\}$ on system $S$ and any tripartite state $\rho_{SE_1E_2}$, then
\begin{equation}
\label{eqn56}
 H_{\min}(X|E_1)+ H_{\max}(Z|E_2) \geq \log\frac{1}{c(X,Z)},
\end{equation}
where $c(X,Z)=\max_{j,k} \| \sqrt{Z_{k}} \sqrt{X_{j}}\|_\infty^2$ (the infinity norm of an operator is its largest singular value). Let us specialize to pure $\rho_{SE_1E_2}$, and combine \eqref{eqn56} with \eqref{eqn53} to obtain
\begin{equation}
\label{eqn57}
 \Ebb_{\fid}^{M_X| \Sbf E_2} + \Ebb_{\max}^{M_Z| \Sbf E_1} \geq \log\frac{1}{c(X,Z)},
\end{equation}
where $\Sbf$ extends $S$ to allow $X$ and $Z$ to be projective. It is interesting that \eqref{eqn57} has nothing to do with uncertainty, and conceptually is just about entanglement, where $\Ebb_{\max}$ has been given an operational meaning as a one-shot entanglement cost \cite{BranDatt2011}, and $\Ebb_{\fid}$ is closely related to the geometric entanglement \cite{StreltsovEtAl2010NJP}. We note that \eqref{eqn57} is stronger than the inequality obtained from replacing $E_1$ and $E_2$ in \eqref{eqn57} with the joint system $E=E_1E_2$, since $\Ebb_{\fid}$ and $\Ebb_{\max}$ are non-increasing under local partial trace, e.g., $\Ebb^{M_X|\Sbf E}_{\fid} \geq \Ebb^{M_X|\Sbf E_2}_{\fid}$. This strengthening of the inequality [i.e, restricting $E$ to its subsystems as in \eqref{eqn57}] corresponds precisely to the strengthening obtained from allowing quantum memory in the uncertainty relation, \eqref{eqn56}.

In addition to proving \eqref{eqn56}, Tomamichel and Renner \cite{TomRen2010} generalized the uncertainty relation to the case of smooth entropies ($\ep \geq 0$):
\begin{equation}
\label{eqn58}
H^{\ep}_{\min}(X|E_1) + H^{\ep}_{\max}(Z|E_2) \geq \log \frac{1}{c(X,Z)}
\end{equation}
This uncertainty relation has been received with significant excitement due to its application in proving the security of QKD even in the non-asymptotic case, where Alice and Bob do only a finite number of measurements \cite{TLGR}. It therefore seems interesting that we can rewrite \eqref{eqn58} in a way that takes on a completely different conceptual meaning. For pure $\rho_{SE_1E_2}$, this uncertainty relation combined with \eqref{eqn54} becomes 
\begin{equation}
\label{eqn59}
{}_{\ep}\Ebb^{M_X|\Sbf E_2}_{\fid} + {}_{\ep}\Ebb^{M_Z |\Sbf  E_1}_{\max} \geq \log \frac{1}{c(X,Z)},
\end{equation}
which reduces to \eqref{eqn57} for the special case of $\ep=0$. Again, we note that \eqref{eqn59} is stronger than the inequality obtained from replacing $E_1$ and $E_2$ in \eqref{eqn59} with the joint system $E=E_1E_2$, since ${}_{\ep}\Ebb_{\fid}$ and ${}_{\ep}\Ebb_{\max}$ are non-increasing under local partial trace (see Appendix~\ref{app44}).

Inequalities of the form of \eqref{eqn57} and \eqref{eqn59} bring to mind the paradigm of entanglement distribution, similar to the discussion in the previous subsection. Here one wishes to establish entanglement between distant locations by, for example, sending a carrier quantum system. A scenario that \eqref{eqn57} and \eqref{eqn59} would be relevant to is the following. Suppose that Alice, Bob, and Charlie, respectively, possess the $S$, $E_1$, and $E_2$ portions of two copies of a tripartite pure state $\ket{\psi}_{SE_1E_2}$, i.e., the overall state is $\ket{\psi}_{SE_1E_2}^{\ot 2}$. Suppose that they do not initially know what the state $\ket{\psi}_{SE_1E_2}$ is, but at the end of the protocol $\ket{\psi}_{SE_1E_2}$ is revealed to them. Perhaps Alice wishes to establish entanglement with Bob or with Charlie. She can perform a premeasurement of observable $X$ on one of her $S$ systems, keep the register $M_X$, and send the resulting $S$ system to Bob. On the other $S$ system, she premeasures observable $Z$, keeps the register $M_Z$, and sends the resulting $S$ system to Charlie.  If $c(X,Z) < 1$, then she will have established entanglement with Bob and/or with Charlie (at least one of the two). This fact is guaranteed by \eqref{eqn57} and \eqref{eqn59}, and of course these inequalities quantitatively bound the amount of entanglement that is established.

\section{Implications and future outlook}\label{sct8}

\subsection{Entanglement complementarity relations}

Our main technical results were given in Section~\ref{sct5}, as theorems stating that a certain class of bipartite states called premeasurement states cause the quantum correlation hierarchy to collapse. However, our most important contribution may be the conceptual insight about the nature of uncertainty and uncertainty relations, discussed in Sections~\ref{sct6} and \ref{sct7}. 

Apparently, many uncertainty relations, which are typically thought of as bounds on our knowledge of incompatible observables, can be reinterpreted as bounds on the entanglement created when incompatible observables are measured. This reinterpretation holds for any EUR for a finite-dimensional quantum system written, e.g., in terms of the Shannon entropy, smooth min-entropy, or smooth max-entropy. 

Perhaps the most important implication is the idea that entanglement creation exhibits complementarity. Of course, researchers are somewhat familiar with the idea because of the so-called ``measurement problem" and the fact that Schrodinger's cat will get produced when a measurement device interacts with a system that is initially in a superposition state. But much of that discussion has been qualitative, whereas, we have shown here that there are \textit{precise} and \textit{general} lower-bounds on entanglement creation during measurement. It seems very interesting that complementarity, the idea that there are certain observables that are incompatible, can be expressed in a manner that has nothing to do with uncertainty. We think it is worthwhile to give these inequalities a name, say, entanglement complementarity relations (ECRs).

Even though each of the ECRs that we have presented in this article is equivalent to some EUR, it seems extremely likely that researchers will find ECRs in the future that have no obvious connection to an EUR. In other words, we think that ECRs are their own class of inequalities, and we believe there is plenty of room to explore them!  This is especially true given that there is a vast zoo of entanglement measures \cite{HHHH09}. More generally, there is a vast zoo of quantum correlation measures \cite{ModiEtAl2011review}, and so we should open to possibly finding complementarity relations for entanglement, discord, and other related correlation measures. 

As discussed in Section~\ref{sct7}, it is possible that these ECRs could be useful for developing strategies to create and distribute entanglement, particularly if they are formulated with entanglement measures that have operational meanings.

\subsection{Implications for quantum correlations}

Here we mention a few more implications of our results in the field of quantum correlations, emphasizing that the connection of EURs to ECRs is the main implication. 

One reason that pure states are so nice is that their entanglement, discord, and relative entropy of quantumness are so easy to calculate - just the entropy of the reduced state. The collapse of the hierarchy for $\textsf{MQ}$ states (which include but go beyond pure states) implies that their entanglement, discord, and relative entropy of quantumness are also quite easy to calculate. Calculating them simply involves calculating a conditional entropy, which, in the von Neumann case, does not involve any optimization process.

Another implication of the hierarchy collapse is that operational meanings get shared. That is, for $\textsf{MQ}$ states, entanglement measures inherit operational meanings of discord \cite{ModiEtAl2011review}, and vice-versa. 

Finally, we note that the entanglement created in premeasurements has been studied previously as a fairly general strategy to quantify discord \cite{PianiEtAl11, StrKamBru11, PianiAdessoPRA.85.040301}. The idea is that a state $\rho_{AB}$ is classical with respect to system $A$ if and only if there exists a premeasurement in some orthonormal basis $W$ on $\HC_A$ that creates no entanglement between the register $M_W$ and the $AB$ system, i.e., if and only if $\rhot_{M_W|AB}\in \textsf{Sep}$. On the other hand, our results (for example, Theorem~\ref{thm4}) imply that the following four conditions are equivalent: $\rhot_{M_W |AB}\in \textsf{Sep} \Leftrightarrow \rhot_{M_W |AB}\in \textsf{CQ} \Leftrightarrow \rhot_{M_W |AB}\in \textsf{QC} \Leftrightarrow \rhot_{M_W |AB}\in \textsf{CC} $. Thus, we have four equivalent classicality conditions. This naturally leads one to think of quantitative measures of the form
\begin{equation}
\label{eqn60}
\DC^{\overrightarrow{A|B}} (\rho_{AB}) = \min_{W\in \WC_A }Q^{M_W |AB}(\rhot_{M_W AB})
\end{equation}
where $\WC_A$ is the set of all orthonormal bases on $\HC_A$, and where $Q$ is any non-negative correlation measure that vanishes only on either $\textsf{Sep}$, $\textsf{CQ}$, $\textsf{QC}$, or $\textsf{CC}$. The quantity $\DC^{\to}$ in \eqref{eqn60} can be thought of as a general one-way discord measure, with the generality of $Q$ giving a slightly more general framework than the case where $Q$ is restricted to be an entanglement measure.

\section{Conclusions}\label{sct9}

We have investigated the hierarchical ordering of quantum correlation measures, in which two-way discord is the broadest kind of quantum correlation (i.e., giving the largest value), and then becoming progressively more narrow (i.e., smaller in value) is one-way discord, then entanglement, and finally coherent information. Each of these four kinds of correlations can be quantified with different measures, for example, we have considered measures related to the von Neumann entropy, measures related to the smooth min- and max-entropies, and measures based on a generic relative entropy. In each case, we find a hierarchical ordering, and furthermore, we find that this hierarchy collapses to a single value for a special class of bipartite states called premeasurement states. In the case of measures related to von Neumann entropy, Section~\ref{sct5.3}, we showed that these states are the \textit{only} states that fully collapse the quantum correlation hierarchy, as in Theorem~\ref{thm4}.

In addition to collapsing the hierarchy, these states are interesting because they can be thought of as being produced from the interaction of a system with a measurement device, schematically shown in Fig.~\ref{fgr1}. Indeed, they have been studied previously in the context of measurement, decoherence, and einselection \cite{ZurekReview}. Maximally correlated states are a special example of premeasurement states, though more generally premeasurement states can be asymmetric with respect to the two subsystems (one-way maximally correlated). In Section~\ref{sct22b} we discussed the relation of premeasurement states to a broader class of states considered in Ref.~\cite{CorDeOFan11}.

Considering the dynamic view that, indeed, the premeasurement state arose from the measurement process in Fig.~\ref{fgr1}, we made the very interesting connection that the quantum correlations of the premeasurement state is precisely connected to the \textit{uncertainty} of the observable being measured. As discussed in Section~\ref{sct6}, this connection holds when the uncertainty is quantified, e.g., with Shannon / von Neumann entropy, smooth min-entropy, or smooth max-entropy. Though we gave a few preliminary results on this idea in \cite{ColesDecDisc2012} (see Theorem 2 of that article), the present article dramatically extends and generalizes this idea. We are left with the realization that uncertainty prior to a measurement implies that entanglement will be created in that measurement (one may need access to the purifying system to see the entanglement), and conversely the production of entanglement implies the lack of certainty about the observable being measured.

This intimate connection between uncertainty and the creation of entanglement (more generally, quantum correlations) has immediate consequences. Researchers have been proving stronger and stronger entropic uncertainty relations (EURs) over the past few decades. But these bounds on our knowledge of incompatible observables can be completely reinterpreted as bounds on the entanglement created when incompatible observables are measured. In Section~\ref{sct7}, we illustrated this idea with a game, where Alice wanted to create and distribute entanglement to Bob, even when she has no idea what state she possesses. Measuring incompatible observables on different copies of her system, and then sending the registers to Bob, is a strategy for Alice to win this game, as guaranteed by ``entanglement complementarity relations", i.e., our reinterpretation of EURs in terms of entanglement creation. Section~\ref{sct7} discussed the reinterpretation of several EURs, including ones allowing for quantum memory \cite{BertaEtAl}, and ones for the smooth min- and max-entropies \cite{TomRen2010}.

Section~\ref{sct8} gives an optimistic future outlook for entanglement complementarity relations (ECRs).  Even though every ECR presented here is linked to some EUR, is it possible to find ECRs that are not linked to some EUR? The present work shows that entanglement creation exhibits the phenomenon of complementarity - if this is a basic principle, then we would expect that there could be a whole class of yet-to-be-discovered inequalities that have nothing to do with uncertainty. These ECRs (or more generally, one can substitute any measure of quantum correlations in place of entanglement, and there is a vast zoo of such measures) offer a new way of capturing the complementarity of quantum mechanics. Exploration into ECRs could inspire strategies to generate and distribute entanglement, and perhaps more importantly, give deeper insight into the complementarity of quantum processes.

\acknowledgments

I thank Shiang Yong Looi, Roger Colbeck, and Marco Piani for helpful discussions, and I especially thank Eric Chitambar for helpful discussions as well as helpful comments on Section~\ref{sct22b}. I note that this work was partly inspired by Refs.~\cite{PianiEtAl11, StrKamBru11, PianiAdessoPRA.85.040301}. Finally, I acknowledge support from the U.S.\ Office of Naval Research.

\appendix

\section{$\textsf{MQ}$ states}\label{app2}

Here we show the equivalence of two alternative definitions of the set of $\textsf{MQ}$ states. Denote the definition given in \eqref{eqn14} and \eqref{eqn6}, respectively, as $ \textsf{MQ}_1 $ and $ \textsf{MQ}_2 $. We will now show that $\textsf{MQ}_1 = \textsf{MQ}_2 $. Consider some $\rho_{AB}\in \textsf{MQ}_1$ given by \eqref{eqn15}, then there exists an orthonormal basis $W=\{\ket{W_j}\}$ on $\HC_A$ whose information is perfectly present in $B$ in the sense that the (unnormalized) conditional density operators on $B$:
\begin{equation}
\label{eqn61}
\tau_{B,j}^W:=\Tr_A(\dya{W_j}\rho_{AB})=X_j\sg_B X_j
\end{equation}
are all orthogonal (i.e., for distinct $j$). In \cite{ColesDecDisc2012}, it was shown that, for pure $\rho_{ABC}$, $\rho_{AC}\in \textsf{CQ}$ iff the information about some orthonormal basis of $\HC_A$ is perfectly present in $B$ [see Eq.~\eqref{eqn62} below for the argument], hence we have shown that $\rho_{AB}\in \textsf{MQ}_2$ and that $\textsf{MQ}_1\subseteq \textsf{MQ}_2$. 

To show the converse that $\textsf{MQ}_2 \subseteq \textsf{MQ}_1$ is significantly more difficult. Suppose $\rho_{AB}\in \textsf{MQ}_2$, i.e., for some purification $\rho_{ABC}$, $\rho_{AC}$ has the form $\sum_j p_j\dya{W_j}\ot \rho_{C,j}$, for some $W\in \WC_A$. It was shown in \cite{ColesDecDisc2012} that for pure $\rho_{ABC}$,
\begin{equation}
\label{eqn62}
H(W|B) = D(\rho_{AC}|| \sum_j \dya{W_j}\rho_{AC}\dya{W_j}),
\end{equation}
where $\dya{W_j}$ is short-hand for $\dya{W_j} \ot \id$. By our assumption that $\rho_{AC}=\sum_j p_j\dya{W_j}\ot \rho_{C,j}$, the right-hand-side of \eqref{eqn62} is zero, and hence $H(W|B)=0$, implying that the $W$ information is perfectly present in $B$, i.e., the conditional density operators
$$\tau_{B,j}^W =\Tr_{AC}(\dya{W_j}\rho_{ABC})$$
are orthogonal for distinct $j$. 

The task now is to show that this condition, $H(W|B)=0$, implies that $\rho_{AB}$ is of the form of \eqref{eqn15}. Our proof of this relies on the conditions for the relative-entropy monotonicity to be satisfied with equality \cite{Petz2003,HaydenEtAl04} and is closely related to the study of minimum uncertainty states of entropic uncertainty relations \cite{ColesYuZwo2011, RenesBoileau}. Let $Z=\{\ket{Z_k}\}$ be the orthonormal basis on $\HC_A$ that is related to $W$ by the Fourier transform, i.e.,
\begin{equation}
\label{eqn63}
\ket{Z_k}=\sum_j \frac{\om^{jk}}{\sqrt{d}}\ket{W_j},\quad \ket{W_j}=\sum_k \frac{\om^{-jk}}{\sqrt{d}}\ket{Z_k},
\end{equation} 
where $d=\dim(\HC_A)$ and $\om=e^{2\pi i/d}$. In general the following uncertainty relation holds for any bipartite state $\rho_{AB}$ \cite{BertaEtAl},
\begin{equation}
\label{eqn64}
H(W|B)+H(Z|B)\geq \log d + H(A|B).
\end{equation} 
Notice that if $H(W|B)=0$, then \eqref{eqn64} is satisfied with equality since, in general, $H(Z|B)\leq \log d + H(A|B)$, implying in this case that $H(Z|B) = \log d + H(A|B)$. Under these conditions, i.e.\ when \eqref{eqn64} is satisfied with equality, we say that $\rho_{AB}$ is a minimum uncertainty state (MUS) of \eqref{eqn64}. Thus, all states for which $H(W|B)=0$ are MUSs of \eqref{eqn64}, so we now proceed to find an analytical form for the MUSs of \eqref{eqn64}. In fact, these MUSs were found previously in \cite{ColesYuZwo2011}, but we repeat some of the discussion here for completeness.

Let $\EC_Z(\cdot) = \sum_k \dya{Z_k}(\cdot )\dya{Z_k}$ be the quantum channel that decoheres in the $Z$ basis. Then (see \cite{ColesYuZwo2011} for more details) the MUSs of \eqref{eqn64} are the states for which
\begin{align}
\label{eqn65}
&D(\rho_{AB} || \sum_j \dya{W_j} \rho_{AB} \dya{W_j} )\notag\\
& = D(\EC_Z(\rho_{AB}) || \EC_Z(\sum_j \dya{W_j} \rho_{AB} \dya{W_j} )),
\end{align}
since, for any $\rho_{AB}$, the left-hand-side of \eqref{eqn65} equals $H(W|B)-H(A|B)$, and the right-hand-side equals $ \log d - H(Z|B)$.  For some quantum channel $\EC$, Petz showed \cite{Petz2003,HaydenEtAl04} that $D(\rho||\sg)=D(\EC(\rho)||\EC(\sg))$ if and only if there exists a quantum channel $\hat \EC$ that undoes the action of $\EC$ on $\rho$ and $\sg$:
\begin{equation}
\label{eqn66}
\hat\EC\EC\rho=\rho,\quad \hat\EC\EC \sg = \sg.
\end{equation} 
The construction given \cite{HaydenEtAl04} for this, defined on the support of $\EC(\sg)$, is
\begin{equation}
\label{eqn67}
\hat\EC(\rho)=\sqrt{\sg}\EC\ad(\EC(\sg)^{-1/2}\rho\EC(\sg)^{-1/2})\sqrt{\sg},
\end{equation} 
which automatically satisfies $\hat\EC\EC \sg = \sg$, so one just needs to solve $\hat\EC\EC\rho=\rho$. To apply this formula to \eqref{eqn65}, we set $\rho=\rho_{AB}$, $\sg= \sum_j \dya{W_j} \rho_{AB} \dya{W_j}$, and $\EC = \EC_Z$. Solving $\rho_{AB} = \hat\EC\EC\rho_{AB}$ gives
\begin{align}
\label{eqn68}
\rho_{AB}=& \sum_{j,j',k} \frac{\om^{(j-j')k}}{d} \dyad{W_j}{W_{j'}} \ot \notag\\
&\sqrt{\tau_{B,j}^W}\rho_B^{-1/2} \tau_{B,k}^Z \rho_B^{-1/2}\sqrt{\tau_{B,j'}^W}
\end{align} 
where $\rho_B=\Tr_A(\rho_{AB})$ and $\tau_{B,k}^Z:= \Tr_A(\dya{Z_k}\rho_{AB})$. Again, the idea is that \eqref{eqn68} is the general form for all MUSs of \eqref{eqn64}, and so we can specialize this formula to the special case where $H(W|B)=0$. This corresponds to all the $\tau_{B,j}^W $ being orthogonal, hence $\rho_B = \bigoplus_j \tau_{B,j}^W$, and
$$X_j:=\sqrt{\tau_{B,j}^W}\rho_B^{-1/2}= \rho_B^{-1/2} \sqrt{\tau_{B,j}^W}$$
is the projector onto the support of $\tau_{B,j}^W$. Thus the $\{X_j\}$ form a set of orthogonal projectors that sum to $\id_B$ provided, if $\rho_B$ is not full rank, then we enlarge one of the projectors, say $X_1$, so that $X_1$ also includes the space orthogonal to $\rho_B$.

Thus, under the condition $H(W|B)=0$, \eqref{eqn68} becomes
\begin{align}
\label{eqn69}
\rho_{AB}&= \sum_{j,j',k} \frac{\om^{(j-j')k}}{d} \dyad{W_j}{W_{j'}} \ot X_j \tau_{B,k}^Z X_{j'}\notag\\
&= \sum_{j,j'} \dyad{W_j}{W_{j'}} \ot X_j \sg_B X_{j'}
\end{align} 
where $\sg_B$ can be defined block-by-block with $X_j\sg_B X_{j'} = \sum_k (1/d) \om^{(j-j')k} \tau_{B,k}^Z$. One can verify that $\sg_B$ is a normalized density operator with $\sg_B = \sum_{j,j'} X_j\sg_B X_{j'} = d \tau_{B,0}^Z$, noting that $\Tr (\tau_{B,0}^Z) = (1/d)$ since $H(W|B)=0$ forces $Z$ to be uniformly distributed by the uncertainty relation. Thus, we have shown that $\rho_{AB}$ has the form of \eqref{eqn15}, so $\rho_{AB}\in \textsf{MQ}_1 $ and $\textsf{MQ}_2 \subseteq \textsf{MQ}_1 $, proving that $\textsf{MQ}_1 = \textsf{MQ}_2 $.

\section{Bures distance for $\textsf{MQ}$ states}\label{app5}

Here we show that the closest separable state to a $\textsf{MQ}$ state is a $\textsf{CC}$ state, as measured by the Bures distance, defined as \cite{BenZyc06}:
\begin{equation}
\label{eqn70}
D_{B}(\rho ,\sg):= \sqrt{2-2F(\rho , \sg)}.
\end{equation}
where $\rho$ and $\sg$ are (normalized) density operators, and $F(\rho , \sg)= \Tr (\sqrt{\rho} \sg\sqrt{\rho})^{1/2} $ is the fidelity.

Consider the following properties of the fidelity. For positive-semidefinite operators $P$ and $Q$, if $\tilde{Q}\geq Q$, then
\begin{equation}
\label{eqn71}
F(P , Q)\leq F(P , \tilde{Q}).
\end{equation}
Also, suppose that $\Pi_{P}$ is a projector onto a space that includes the support of $P$, then
\begin{equation}
\label{eqn72}
F(P ,Q)= F(P , \Pi_{P} Q\Pi_{P}).
\end{equation}
Now consider a general $\textsf{MQ}$ state, which is of the form $\tilde{\rho}_{M_XS}=V_X\rho_S V_X\ad$, for some density operator $\rho_S$ on $S$ and some premeasurement isometry $V_X: \HC_S\to\HC_{M_XS}$, which has the form given in \eqref{eqn11}. In what follows, we first use \eqref{eqn71}, noting that if $\sg_{M_XS}\in\textsf{Sep}$, then $\id \ot \sg_S\geq \sg_{M_XS}$, with $\sg_{S} =\Tr_{M_X}(\sg_{M_XS})$. We then use \eqref{eqn72}, noting that $V_X V_X\ad$ projects onto a space that includes the support of $\tilde{\rho}_{M_XS}$. For some $\sg_{M_XS} \in \textsf{Sep}$, we find:
\begin{align}
\label{eqn73}
F(\tilde{\rho}_{M_XS} ,\sg_{M_XS})&\leq F(\tilde{\rho}_{M_XS} , \id \ot \sg_{S}) \notag\\
&= F(\tilde{\rho}_{M_XS} , V_X V_X\ad (\id \ot \sg_{S}) V_X V_X\ad )\notag\\
&= F(\tilde{\rho}_{M_XS} ,\al_{M_XS} ),
\end{align}
where
\begin{align}
\label{eqn74}
\al_{M_XS} &:= V_X \sum_j X_j \sg_{S} X_j V_X\ad\notag\\
&= \sum_j \dya{j}\ot X_j \sg_{S} X_j  
\end{align}
is a $\textsf{CC}$ state, which can be verified by expanding the $X_j \sg_{S} X_j$ blocks in their eigenbasis. Therefore, \eqref{eqn73} shows that, for any separable state $\sg_{M_XS}$, there is a $\textsf{CC}$ state $\al_{M_XS}$ that is closer to $\tilde{\rho}_{M_XS}$, according to the fidelity. Since $D_B$ varies monotonically with $F$, then this statement also holds for $D_B$. Thus, the closest $\textsf{Sep}$ state, according to $D_B$, is a $\textsf{CC}$ state.

\section{Proof of Lemma~\ref{thm3}}\label{app4}

\begin{proof}
For \eqref{eqn36}, $I_c\leq \Ebb_D$ was shown in \cite{DevWin05} and the fact that a bit of secret key can be obtained from an e-bit implies $\Ebb_D\leq K_D$. Now it is obvious from \eqref{eqn4} that $E_R\leq \Dl^{\to} \leq \Dl^{\leftrightarrow}$, and the additivity of the von Neumann relative entropy implies that regularisation cannot increase these measures: $E_{R,\infty} \leq E_R$, $\Dl^{\to}_{\infty}\leq \Dl^{\to}$, and $\Dl^{\leftrightarrow}_{\infty} \leq \Dl^{\leftrightarrow}$. Combining this with a result from \cite{HHHOprl05} that $K_D\leq E_{R,\infty}$ implies \eqref{eqn36}.

For \eqref{eqn37}, we have from \eqref{eqn35} that $E_{R,\infty}\leq \Dl^{\to}_{\infty}=\dl^{\to}_{\infty}\leq \dl^{\to}$, where the last inequality follows from the additivity of the mutual information. Also, the relation $\dl^{\to}\leq \dl^{\leftrightarrow}$ follows from the Holevo bound \cite{NieChu00}. The right-most inequality in \eqref{eqn37} goes as follows. Note that $\dl^{\leftrightarrow}$ is smaller than the case where the minimization is performed over all rank-one projective measurements $\WC$ and $\WC'$ on $A$ and $B$, respectively, so
\begin{align}
\dl^{\overleftrightarrow{A|B}}&\leq \min_{\WC, \WC'} \{ I(\rho_{AB}) - I[(\WC \ot \WC')(\rho_{AB})] \} \notag\\
&\leq \min_{\WC, \WC'} \{ H[ (\WC \ot \WC' ) (\rho_{AB})] -  H(\rho_{AB}) \} = \Dl^{\overleftrightarrow{A|B}}. \notag
\end{align}

For \eqref{eqn38}, simply apply the above arguments, such as Eq.~\eqref{eqn37}, to the state $\rho_{AB}^{\ot N}$ for $N\to \infty$.

\end{proof}

\section{Smooth measures}\label{app44}

\subsection{Subnormalized states}

Let $\textsf{P}(\HC)$ denote the set of positive semi-definite operators on Hilbert space $\HC$. Let $\textsf{S}_{\le}(\HC)$ and $\textsf{S}_{=}(\HC)$, respectively, denote the sets of subnormalized and normalized positive operators on $\HC$, i.e., 
\begin{align}
\textsf{S}_{\le}(\HC) & = \{\sg \in \textsf{P}(\HC): \Tr(\sg) \leq 1 \},\notag\\
\textsf{S}_{=}(\HC) &= \{\sg \in \textsf{P}(\HC): \Tr(\sg) = 1 \}.\notag
\end{align}
Sometimes we may drop the explicit dependence on the Hilbert space for $\textsf{S}_{\le}$ and $\textsf{S}_{=} $ when the space is obvious.

It also useful to generalize the notion of $\textsf{MQ}$ states to subnormalized states. We denote this broader set as $\textsf{MQ}_{\le}$, containing all states of the form given in \eqref{eqn14} but allowing $\sg_B\in \textsf{S}_{\le}(\HC_B)$ to be subnormalized, or equivalently, defined by \eqref{eqn6} but allowing $\rho_{AC}$ to be subnormalized.

\subsection{Purified distance and $\ep$-balls}

Smooth measures involve optimizing over a ball of states within some chosen distance $\ep$ from the state of interest. These balls of states are called $\ep$-balls, and the distance measure of choice for constructing them is the purified distance \cite{TomColRen10}, which ensures that the $\ep$-balls are, to some degree, invariant under purifications or extensions (e.g., see Lemma~\ref{thm9}). The definitions and lemmas in this subsection are mostly due to the work of Tomamichel, Colbeck, and Renner \cite{TomColRen10}. They note that the purified distance between $\rho \in \textsf{S}_{\le}$ and $\sg \in \textsf{S}_{\le}$ can be written as

\begin{equation}
P(\rho, \sg) = \sqrt{1-\Fbar(\rho, \sg)^2},
\end{equation}
where $\Fbar$ is a generalized fidelity, 
$$\Fbar(\rho, \sg):=F(\rho,\sg)+\sqrt{(1-\Tr \rho)(1-\Tr \sg)}$$
with the standard fidelity given by
$$F(\rho,\sg) = \Tr [(\sqrt{\sg}\rho \sqrt{\sg})^{1/2}]. $$ 
Several useful properties of the purified distance are worked out in \cite{TomColRen10}. For example, they give the following lemma.
\begin{lemma}\label{thm7}
The purified distance is non-increasing under trace-nonincreasing completely positive maps (TNICPMs). Since a projector $\Pi$ gives rise to a TNICPM of the form $\rho \to \Pi\rho \Pi$, we have that
$$P(\rho , \sg) \geq P( \Pi \rho \Pi , \Pi \sg \Pi).$$\openbox
\end{lemma}

Now let us define the $\ep$-ball around $\rho\in\textsf{S}_{\le}$,
$$\BC^{\ep}(\rho):=\{\sg \in \textsf{S}_{\le}: P(\rho, \sg)\leq \ep \}.$$
The ball grows monotonically with $\ep$ and $\BC^0(\rho) = \{\rho \}$.

Also, Lemma~\ref{thm7} implies that if $\sg \in \BC^{\ep}(\rho)$, then $\Pi \sg \Pi \in \BC^{\ep}(\rho)$ if $\Pi$ projects onto a space that includes the support of $\rho$. This fact is helpful when considering a subset of the $\ep$-ball that includes only those states confined to a particular subspace $\Pi$ that includes the support $\rho$, defined as follows 
$$\BC^{\ep}_{\Pi}(\rho):=\{\sg \in \textsf{S}_{\le}: P(\rho, \sg)\leq \ep , \Pi  \sg \Pi =\sg \}. $$
Clearly $\BC^{\ep}_{\Pi}(\rho)\subseteq \BC^{\ep}(\rho)$, and setting $\Pi$ to the identity recovers the full ball, $\BC^{\ep}_{\id}(\rho)=\BC^{\ep}(\rho)$.

It will also useful to define another subset of the $\ep$-ball in the special case where $\rho$ is pure, and that is a ball of pure states \cite{TomColRen10}:
 $$\BC^{\ep}_p(\rho):=\{\sg \in \textsf{S}_{\le}: P(\rho, \sg)\leq \ep,\text{rank}(\sg)=1 \}. $$
Again, $\BC^{\ep}_p(\rho) \subseteq \BC^{\ep}(\rho)$. In fact it is helpful to combine the notions of $\BC^{\ep}_{\Pi}(\rho)$ and $\BC^{\ep}_p(\rho)$ as follows:
 $$\BC^{\ep}_{p,\Pi}(\rho):=\{\sg \in \textsf{S}_{\le}: P(\rho, \sg)\leq \ep,\text{rank}(\sg)=1, \Pi  \sg \Pi =\sg \},$$
assuming $\rho$ is pure and $\Pi  \rho \Pi =\rho$.

The following two lemmas are from \cite{TomColRen10}.
\begin{lemma} \label{thm8}
Let $\rho, \tau \in \textsf{S}_{\le}(\HC)$ and let $\phi_{\rho}\in \textsf{S}_{\le}(\HC\ot\HC ' )$ with $\dim(\HC)\leq \dim(\HC')$ be a purification of $\rho$, then there exists a purification of $\tau$, $\phi_{\tau}\in \textsf{S}_{\le}(\HC\ot\HC ' )$, such that $P(\phi_{\rho}, \phi_{\tau}) = P(\rho,\tau)$.
\end{lemma}
\begin{lemma} \label{thm9}
Let $\rho \in \textsf{S}_{\le}(\HC)$ and let $\phi_{\rho}\in \textsf{S}_{\le}(\HC\ot\HC ' )$ be a purification of $\rho$, then
$$\BC^{\ep}(\rho) \supseteq \{\sg\in \textsf{S}_{\le}(\HC): \exists \phi_{\sg}\in \BC^{\ep}_p(\phi_{\rho}) , \sg =\Tr_{\HC '}(\phi_{\sg}) \}$$
with the two sets being identical if $\dim(\HC)\leq \dim(\HC')$.
\end{lemma}

We will need a generalization of Lemma~\ref{thm9}. 
\begin{lemma} \label{thm10}
Let $\rho \in \textsf{S}_{\le}(\HC)$ and let $\phi_{\rho}\in \textsf{S}_{\le}(\HC\ot\HC ' )$ be a purification of $\rho$, let $\Pi$ be a projector such that $\Pi\rho\Pi = \rho$, then
\begin{equation}
\label{eqn77}
\BC^{\ep}_{\Pi}(\rho) \supseteq \{\sg\in \textsf{S}_{\le}(\HC): \exists \phi_{\sg}\in \BC^{\ep}_{p,\Pi}(\phi_{\rho}) , \sg =\Tr_{\HC '}(\phi_{\sg}) \}
\end{equation}
with the two sets being identical if $\dim(\HC)\leq \dim(\HC')$.
\end{lemma}
\begin{proof}
Note that if $\sg\in \textsf{S}_{\le}(\HC)$, if $\Pi$ is a projector on $\HC$, and if $\phi_{\sg}\in \textsf{S}_{\le}(\HC\ot \HC ') $ is a purification of $\sg$, then $\Pi \sg\Pi = \sg$ if and only if $\Pi \phi_{\sg}\Pi = \phi_{\sg} $. So the set on the right side of \eqref{eqn77} will be contained in the $\Pi$ subspace and will be within $\ep$ of $\rho$ since the purified distance is non-increasing under partial trace, so this proves \eqref{eqn77}. The equality for $\dim(\HC)\leq \dim(\HC')$ follows from Lemma~\ref{thm8}.
\end{proof}

\subsection{Additional properties of $D_{\max}$ and $D_{\fid}$}

In addition to Properties~\ref{a} and \ref{b}, it is useful note the following properties of $D_{\max}$ and $D_{\fid}$.

For positive operators $P$ and $Q$, if $P' \geq P$, then
\begin{align}
\label{eqn78}D_{\max}(P || Q ) &\leq D_{\max}(P' ||Q) \\
\label{eqn79}D_{\fid}(P || Q ) &\geq D_{\fid}(P' ||Q )
\end{align}
Let $\Pi_P$ be a projector that includes the support of $P$, and define $\Pi_Q$ analogously, then
\begin{align}
\label{eqn80}D_{\fid}(P || Q ) &= D_{\fid}(P || \Pi_P Q \Pi_P)\notag\\
&= D_{\fid}(\Pi_Q P \Pi_Q || Q )
\end{align}

Let $\Pi$ be any projector, then
\begin{align}
\label{eqn81}D_{\max}(P || Q ) &\geq D_{\max}(\Pi  P \Pi || \Pi Q \Pi )
\end{align}
\begin{proof}
Consider the quantum channel $\FC(\cdot ) := \Pi(\cdot )\Pi +(\id - \Pi)(\cdot )(\id - \Pi)$. We have
\begin{align}
D_{\max}(P || Q) &\geq D_{\max}(\FC(P) || \FC( Q))\notag \\
& \geq D_{\max}(\Pi P \Pi || \Pi Q \Pi) \notag
\end{align}
where the second line follows by first invoking \eqref{eqn78} with $\FC(P)\geq  \Pi P \Pi $, and then invoking Property~\ref{b}.
\end{proof}

\subsection{Proof of Lemma~\ref{thm5}}

In Lemma~\ref{thm1}, we proved a correlation hierarchy for the the min- and max-entropies (among others). The proof was for normalized states $\rho_{AB}$, but the exact same proof applies to subnormalized states. Hence we have the following Lemma.
\begin{lemma}\label{thm11}
For any bipartite state $\rho_{AB}\in \textsf{S}_{\le}$, 
\begin{align}
\label{eqn82}&-H_{\min}(A|B) \leq  \Ebb^{A|B}_{\max} \leq  \Dl^{\overrightarrow{A|B}}_{\max}, \Dl^{\overrightarrow{B|A}}_{\max} \leq \Dl^{\overleftrightarrow{A|B}}_{\max},\\ 
\label{eqn83}&-H_{\max}(A|B) \leq  \Ebb^{A|B}_{\fid} \leq  \Dl^{\overrightarrow{A|B}}_{\fid}, \Dl^{\overrightarrow{B|A}}_{\fid} \leq  \Dl^{\overleftrightarrow{A|B}}_{\fid}.
\end{align}
\end{lemma}

Now since Lemma~\ref{thm11} applies to each state in the ball, $\sg_{AB}\in \BC^{\ep}(\rho_{AB})$, Lemma~\ref{thm5} follows as a direct corollary.

\subsection{Proof of Theorem~\ref{thm6}}

Here we prove the collapse of the smooth correlation hierarchy for premeasurement states. It is helpful to note the following lemma, which extends Theorem~\ref{thm2} to subnormalized states. The proof is exactly the same as that given for Theorem~\ref{thm2}, i.e., the proof of Theorem~\ref{thm2} did not rely on the normalization of the state.
\begin{lemma}\label{thm12}
For any premeasurement state $\rhot_{M_XS}\in \textsf{MQ}_{\le}$,
\begin{align}
\label{eqn84}&-H_{\min}(M_X|S) = \Ebb^{M_X|S}_{\max} = \Dl^{\overrightarrow{M_X|S}}_{\max}= \Dl^{\overrightarrow{S|M_X}}_{\max} = \Dl^{\overleftrightarrow{M_X|S}}_{\max},\\
\label{eqn85}&-H_{\max}(M_X|S) = \Ebb^{M_X|S}_{\fid} = \Dl^{\overrightarrow{M_X|S}}_{\fid}= \Dl^{\overrightarrow{S|M_X}}_{\fid} = \Dl^{\overleftrightarrow{M_X|S}}_{\fid}.
\end{align}\end{lemma}

In what follows, we make use of $\BC^{\ep}_{\Pi}(\rhot_{M_XS})$ where $\Pi =V_X V_X\ad$ includes the support of $\rhot_{M_XS} = V_X\rho_S V_X\ad$, so it is helpful to state the following lemma.
\begin{lemma} \label{thm13}
For $\rhot_{M_XS} = V_X\rho_S V_X\ad \in \textsf{MQ}$ and $\Pi =V_X V_X\ad$, the ball $\BC^{\ep}_{\Pi}(\tilde{\rho}_{M_XS})$ only contains $\textsf{MQ}_{\le}$ states, of the form $V_X \tau_S V_X\ad$ for some $\tau_S\in \textsf{S}_{\le}(\HC_S)$.
\end{lemma}
\begin{proof}
By definition, all states in $\BC^{\ep}_{\Pi}(\tilde{\rho}_{M_XS})$ are of the form $\sg_{M_XS} = V_X V_X\ad \sg_{M_XS} V_X V_X\ad$ and hence of the form $V_X \tau_S V_X\ad$ where $\tau_S = V_X\ad \sg_{M_XS} V_X \in \textsf{S}_{\le}(\HC_S)$.
\end{proof}

Now we are ready to prove Theorem~\ref{thm6}. Let $\Pi = V_XV_X\ad$ in what follows. We first show the proof of \eqref{eqn46}. Let $\sgbar \in\textsf{S}_{\le}(\HC_{M_XS})$ and $\taubar \in \textsf{S}_{=}(\HC_S)$ be the two states that achieve the optimization in $H^{\ep}_{\min}(M_X|S)_{\rhot}$, i.e., let $- H^{\ep}_{\min}(M_X|S)_{\rhot} = D_{\min}(\sgbar || \id \ot \taubar )$, where $\rhot$ is short-hand for $\rhot_{M_XS}$. Then we have
\begin{align}
\label{eqn86}- H^{\ep}_{\min}(M_X|S)_{\rhot} &= D_{\min}(\sgbar || \id \ot \taubar )  \\
\label{eqn87}& \geq D_{\min}(\Pi \sgbar \Pi || \Pi (\id \ot \taubar) \Pi) \\
\label{eqn88}&= D_{\min}(\Pi \sgbar \Pi || \id \ot \sum_j X_j \taubar X_j)\\
\label{eqn89}&\geq  \min_{\sg\in\BC^{\ep}_{\Pi}(\rhot)}  [-H_{\min}(M_X|S)_{\sg}] \\
\label{eqn90}&=  \min_{\sg\in\BC^{\ep}_{\Pi}(\rhot)}  \Dl^{\overleftrightarrow{M_X|S}}_{\max}(\sg) \\
\label{eqn91}&\geq {}_{\ep}\Dl^{\overleftrightarrow{M_X|S}}_{\max}(\rhot)
 \end{align}
Equation~\eqref{eqn87} invoked \eqref{eqn81}, \eqref{eqn88} invoked Property~\ref{b}, \eqref{eqn90} invoked Lemmas~\ref{thm12} and \ref{thm13}, and \eqref{eqn91} notes that $\BC^{\ep}_{\Pi}(\rhot) \subset \BC^{\ep}(\rhot)$. Now note that \eqref{eqn44} gave an inequality in the reverse direction, so the inequalities must be equalities and the hierarchy in \eqref{eqn44} must collapse.

For the proof of \eqref{eqn47}, let us define $\textsf{CC}_{\Pi}\subset \textsf{CC}$ as the set $\{\tau\in \textsf{CC}: \Pi \tau \Pi = \tau\}$, i.e., only those $\textsf{CC}$ states that live in the subspace $\Pi$. Then we have
\begin{align}
\label{eqn92}- H^{\ep}_{\max}(M_X|S)_{\rhot} &\geq \max_{\sg\in\BC^{\ep}_{\Pi}(\rhot)}[-H_{\max}(M_X|S)_{\sg}]  \\
\label{eqn93}&= \max_{\sg\in\BC^{\ep}_{\Pi}(\rhot)}  \Dl^{\overleftrightarrow{M_X|S}}_{\fid}(\sg)\\
\label{eqn94}&= \max_{\sg\in\BC^{\ep}_{\Pi}(\rhot)}  \min_{\tau\in \textsf{CC}} D_{\fid}(\sg || \tau )\\
\label{eqn95}&= \max_{\sg\in\BC^{\ep}_{\Pi}(\rhot)}  \min_{\tau\in \textsf{CC}_{\Pi}} D_{\fid}(\sg || \tau )\\
\label{eqn96}&= \max_{\sg\in\BC^{\ep}(\rhot)}  \min_{\tau\in \textsf{CC}_{\Pi}} D_{\fid}(\sg || \tau )\\
\label{eqn97}&\geq \max_{\sg\in\BC^{\ep}(\rhot)}  \min_{\tau\in \textsf{CC}} D_{\fid}(\sg || \tau )\\
\label{eqn98}&= {}_{\ep}\Dl^{\overleftrightarrow{M_X|S}}_{\fid}({\rhot})
 \end{align}
Equation~\eqref{eqn93} invoked Lemmas~\ref{thm12} and \ref{thm13}, \eqref{eqn95} follows from \eqref{eqn73} and the surrounding discussion, \eqref{eqn96} invoked \eqref{eqn80}, and \eqref{eqn97} used $\textsf{CC}_{\Pi}\subset \textsf{CC}$. Again, note that \eqref{eqn45} gave an inequality in the reverse direction, so the inequalities must be equalities and the hierarchy in \eqref{eqn45} must collapse. This completes the proof.

As an aside, we note that, because the above inequalities must be equalities, the optimization in the smooth min- and max-entropy of a premeasurement state can be restricted to the ball $\BC^{\ep}_{\Pi}(\rhot_{M_XS})$, as in Eqs.~\eqref{eqn89} and \eqref{eqn92}.

\subsection{Properties of $_{\ep}\Ebb_{\max}$ and $_{\ep}\Ebb_{\fid}$}

Here we note a few useful properties of $_{\ep}\Ebb_{\max}$ and $_{\ep}\Ebb_{\fid}$. In particular, $_{\ep}\Ebb_{\max}$ is non-increasing under LOCC, $_{\ep}\Ebb_{\fid}$ is non-increasing under local quantum channels, and both $_{\ep}\Ebb_{\max}$ and $_{\ep}\Ebb_{\fid}$ are invariant under local isometries.
\begin{lemma}
Let $\rho_{AB}\in \textsf{S}_{=}(\HC_{AB})$,

(i)  Let $\Lm$ be an LOCC operation, denote $\rho_{A'B'}=\Lm (\rho_{AB})$, then
\begin{align}
{}_{\ep}\Ebb_{\max}^{A|B}(\rho_{AB}) &\geq {}_{\ep}\Ebb_{\max}^{A'|B'}(\rho_{A'B'})
\end{align}

(ii) Let $\EC_A:\HC_A\to\HC_{A'}$ and $\EC_B:\HC_B\to\HC_{B'}$ be local quantum channels on $A$ and $B$ respectively, denote $\rho_{A'B'}=(\EC_A\ot \EC_B) (\rho_{AB})$, then
\begin{align}
{}_{\ep}\Ebb_{\fid}^{A|B}(\rho_{AB}) &\geq {}_{\ep}\Ebb_{\fid}^{A'|B'}(\rho_{A'B'})
\end{align}

(iii) Let $V_A:\HC_A\to\HC_{A'}$ and $V_B:\HC_B\to\HC_{B'}$ be local isometries on $A$ and $B$ respectively, denote $\rho_{A'B'}=(V_A\ot V_B) \rho_{AB} (V_A\ad \ot V_B \ad )$, then
\begin{align}
{}_{\ep}\Ebb_{\max}^{A|B}(\rho_{AB}) &= {}_{\ep}\Ebb_{\max}^{A'|B'}(\rho_{A'B'})\\
{}_{\ep}\Ebb_{\fid}^{A|B}(\rho_{AB}) &= {}_{\ep}\Ebb_{\fid}^{A'|B'}(\rho_{A'B'})
\end{align}

\end{lemma}
\begin{proof}

(i) \begin{align}
{}_{\ep}\Ebb_{\max}^{A|B}(\rho_{AB}) &= \min_{\sg \in \BC^{\ep}(\rho_{AB})} \Ebb_{\max}^{A|B}(\sg)\notag\\
& \geq \min_{\sg \in \BC^{\ep}(\rho_{AB})} \Ebb_{\max}^{A|B}(\Lm(\sg))\notag\\
& \geq {}_{\ep}\Ebb_{\max}^{A'|B'}(\rho_{A'B'})\notag
\end{align}
where the third line used the fact if $\sg \in \BC^{\ep}(\rho_{AB})$ then $\Lm(\sg) \in \BC^{\ep}(\Lm(\rho_{AB}))$ due to Lemma~\ref{thm7}.

(ii) Using the Stinespring dilation, write $\EC_A(\cdot) = \Tr_{E_A}[V_A(\cdot)V_A\ad ]$ and $\EC_B(\cdot) = \Tr_{E_B}[V_B(\cdot)V_B\ad ]$ where $E_A$ and $E_B$ are ancillas and $V_A:\HC_A \to \HC_{A'E_A}$ and $V_B:\HC_B \to \HC_{B'E_B}$ are local isometries. Define $\rho_{E_A A'B' E_B}:= (V_A\ot V_B)\rho_{AB}(V_A\ad \ot V_B\ad)$ and note that $\rho_{A'B' } = \Tr_{E_AE_B}(\rho_{E_A A'B' E_B}) = (\EC_A\ot \EC_B)(\rho_{AB})$. Also define $\Pi := V_A V_A\ad \ot V_BV_B\ad$, and denote $\textsf{Sep}_{\Pi}$ as the set of normalized separable states that live only in the subspace defined by $\Pi$. Then 
\begin{align}
{}_{\ep}\Ebb_{\fid}^{A|B}(\rho_{AB}) &= \max_{\sg \in \BC^{\ep}(\rho_{AB})} \min_{\tau\in \textsf{Sep}} D_{\fid}(\sg ||\tau ) \notag\\
&=\max_{\sg \in \BC^{\ep}_{\Pi}(\rho_{E_AA'B'E_B})} \min_{\tau\in \textsf{Sep}_{\Pi}} D_{\fid}(\sg ||\tau ) \notag\\
&=\max_{\sg \in \BC^{\ep}(\rho_{E_AA'B'E_B})} \min_{\tau\in \textsf{Sep}_{\Pi}} D_{\fid}(\sg ||\tau ) \notag\\
&\geq \max_{\sg \in \BC^{\ep}(\rho_{E_AA'B'E_B})} \min_{\tau\in \textsf{Sep}} D_{\fid}(\sg ||\tau ) \notag\\
&= \max_{\sg \in \BC^{\ep}(\rho_{E_AA'B'E_B})} E^{E_AA' | B'E_B}_{\fid}(\sg) \notag\\
&\geq \max_{\sg \in \BC^{\ep}(\rho_{A'B'})} E^{A' | B'}_{\fid}(\sg) = {}_{\ep}\Ebb_{\fid}^{A'|B'}(\rho_{A'B'})\notag
\end{align}
The third line follows from \eqref{eqn80} and Lemma~\ref{thm7}. The last line follows from Lemma~\ref{thm7} and the fact that $\Ebb_{\fid}$ is non-increasing under local partial traces.

(iii) This follows from parts (i) and (ii) of this Lemma, by invoking the fact that the entanglement measure is non-increasing under local quantum channels, \textit{twice} in succession. That is, invoke it first with the channel that applies the local isometries $V_A \ot V_B$, and invoke it again with the channel that undoes these local isometries to obtain $\Ebb(\rho_{AB}) \geq \Ebb[(V_A \ot V_B) \rho_{AB} ( V_A\ad \ot V_B\ad)] \geq \Ebb(\rho_{AB})$. Hence the inequalities are equalities.

\end{proof}

\section{Additivity of $\Ebb_D$ for $\textsf{MQ}$ states}\label{app1}

In general, $\Ebb_D$ is not additive \cite{ShorEtAl2001, ShorEtAl2003}, i.e., there exist states $\rho$ and $\sg$ for which $\Ebb_D(\rho\ot \sg)\neq \Ebb_D(\rho) +\Ebb_D(\sg)$. However, in the special case, e.g., when $\rho$ and $\sg$ are $\textsf{MQ}$ states, $\Ebb_D$ is additive. The basic idea is that if $\rho \in \textsf{MQ}$ and $\sg\in \textsf{MQ}$, then $(\rho \ot \sg) \in \textsf{MQ} $, and hence from Theorem~\ref{thm4}, $\Ebb_D(\rho\ot \sg)$ can be written as a conditional von Neumann entropy, and such entropies are additive, which in turn implies the additivity of $\Ebb_D$. This argument of course applies to the state $\rhot_{M_XS}\ot \rhot_{M_YS}\ot \rhot_{M_ZS}$, which is the state considered in the entanglement distillation game in Section~\ref{sct7.2}.

\bibliographystyle{naturemag}
\bibliography{EntanglementUR}

\end{document}